\documentclass[a4paper,reqno,11pt]{article}

\usepackage[hmargin=2.3cm,vmargin=1.9cm]{geometry}

\usepackage{amsmath,amssymb,amsthm,amssymb,amsmath,mathrsfs,stmaryrd,lscape,rotating,adjustbox,
			mathdots,mathtools,color,graphicx,framed}

\usepackage[colorlinks=true, pdfstartview=FitV, urlcolor=blue, citecolor=red, linkcolor=blue]{hyperref}

\definecolor{shadecolor}{rgb}{0.9, 0.9, 0.86}

\def \wt{\widetilde}
\def \wh{\widehat}
\def\Hom{\mathrm{Hom}}
\def\C{\mathbb{C}}
\def\R{\mathbb{R}}
\def\Z{\mathbb{Z}}
\def\1{\mathbf{1}}
\def\b{\beta}
\def\a{\alpha}
\def\L{\Lambda}
\def\d{\mathrm d}
\def\O{\mathcal{O}}
\def\e{\epsilon}
\def\i{\mathrm{i}}
\def\pa{\partial}
\def\s{\sigma}
\def\Mgn{\overline{\mathcal{M}}_{g,n}}
\def\ll{\lambda}
\def\Y{\mathscr{P}}

\def\be{\begin{equation}}
\def\ee{\end{equation}}

\newtheorem{theorem}{Theorem}[section]
\newtheorem{conjecture}{Conjecture}[section]

\newtheorem{lemma}[theorem]{Lemma}
\newtheorem{remark}[theorem]{Remark}
\newtheorem{proposition}[theorem]{Proposition} 
\newtheorem{corollary}[theorem]{Corollary}

\begin{document}

\numberwithin{equation}{section}

\title{On the spectral problem of the quantum KdV hierarchy}
\date{}

\author{Giulio Ruzza\footnote{\text{IRMP, UCLouvain, 1348 Louvain-la-Neuve, Belgium;} \texttt{giulio.ruzza@uclouvain.be}}, \quad Di Yang\footnote{\text{School of Mathematical Sciences, USTC, 230026 Hefei, P.R. China;} \texttt{diyang@ustc.edu.cn}}}

\maketitle

\begin{abstract}
The spectral problem for the quantum dispersionless Korteweg--de Vries (KdV) hierarchy, aka the quantum Hopf hierarchy, is solved by Dubrovin.
In this article, following Dubrovin, we study Buryak--Rossi's quantum KdV hierarchy. 
In particular, we prove a symmetry property and a non-degeneracy property for the quantum KdV Hamiltonians.
On the basis of this we construct a complete set of common eigenvectors.
The analysis underlying this spectral problem implies certain vanishing identities for combinations of characters of the symmetric group.
We also comment on the geometry of the spectral curves of the quantum KdV hierarchy and we give a representation of 
the quantum dispersionless KdV Hamiltonians in terms of multiplication operators in the class algebra of the symmetric group.
\end{abstract}

\section{Introduction}

Let $\L:=\C[q_1,q_2,q_3,\dots]$ denote the ring of polynomials of the indeterminates $q_1,q_2,q_3,\dots$ with complex coefficients.
By a \emph{quantum integrable hierarchy} we mean a sequence of pairwise-commuting 
linear operators $\left(H_m\right)_{m\in I}$ on $\L\otimes R$ or on some completion of it,
satisfying certain non-degeneracy property (cf. e.g.~\cite{Bazhanov1,Bazhanov2,Bazhanov3,BSTV2016,BR2016a,D2016}).
Here $I$ is an infinite index set and $R$ is a certain ring.

To illustrate this subject and the context of this paper, let us start by recalling the definition of the classical KdV hierarchy 
and its Hamiltonian structures.
Denote by $(\mathcal{A}_u,\partial_x)$ the differential polynomial ring of~$u$, namely an element of~$\mathcal{A}_u$ is a polynomial of 
$u_x,u_{xx},u_{xxx},\dots$ whose coefficients are formal power series in~$u$ with complex coefficients. 
Define a sequence of elements $h_{-1}^{\sf cl},h_0^{\sf cl},h_1^{\sf cl},h_2^{\sf cl}$, \dots in~$\mathcal{A}_u[\e^2]$
by the \emph{Lenard-Magri recursion}:
\begin{align}
& (2m+3) \partial_x h_m^{\sf cl} = \biggl(2 u\partial_x + u_x + \frac{\e^2}4 \partial_x^3 \biggr) h_{m-1}^{\sf cl} , \quad m\geq 0,\\
& h_{-1}^{\sf cl} = u, \quad h_m^{\sf cl}|_{u=u_x=u_{xx} =\dots =0} =0,
\end{align}
as well as
\be
\sum_{i\geq 1} i u_{ix} \frac{\pa h_m^{\sf cl}}{\pa u_{ix}} - \e \frac{\pa h_m^{\sf cl}}{\pa \e} =0, \quad m\geq 0.
\ee
Here, $u_{ix} := \pa^i_x(u)$, $i\geq 0$. Explicitly,
\be
h_0^{\sf cl}=\frac{u^2}2+\frac{\e^2}{12} u_{xx}\,, 
\quad h_1^{\sf cl} = \frac{u^3}{6} + \e^2 \Bigl(\frac1{24} u_x^2+\frac1{12} u u_{xx} \Bigr) + \frac{\e^4}{240} u_{xxxx}\,,
\ee
etc. For an element $a\in \mathcal{A}_u[\e]$, we denote by $\overline{a}$ 
the projection of~$a$ onto $\mathcal{A}_u[\e] / \partial_x \mathcal{A}_u[\e]$. 
Elements in $\mathcal{A}_u[\e] / \partial_x \mathcal{A}_u[\e]$ can be called 
{\it local functionals}.  The \emph{classical KdV hierarchy} is the following system of commuting Hamiltonian evolutionary PDEs:
\begin{align}
\frac{\partial u}{\partial t_m} = \mathcal P_1 \Biggl(\frac{\delta \overline{h_m^{\sf cl}}}{\delta u(x)} \Biggr) 
=  \frac1{2m+1}\mathcal P_2 \Biggl(\frac{\delta \overline{h_{m-1}^{\sf cl}}}{\delta u(x)} \Biggr) , \quad m\geq 0, \label{KdV}
\end{align}
where $\frac{\delta}{\delta u(x)}$ denotes the variational derivative, and 
$\mathcal P_1,\mathcal P_2$ are two Hamiltonian operators given by 
\be
\mathcal P_1 =\partial_x,  \quad \mathcal P_2=2 u\partial_x + u_x + \frac{1}4 \e^2 \partial_x^3,
\ee
respectively. The local functionals $\overline{h_m^{\sf cl}}$, $m\geq -1$ serve as the Hamiltonians.
The Hamiltonian operators $\mathcal P_a$, $a=1,2$,
define two Poisson structures on the space of local 
functionals~$\mathcal{A}_u[\e] / \partial_x \mathcal{A}_u[\e]$ in the following way: 
for any $\overline{f}, \overline{g}\in \mathcal{A}_u[\e] / \partial_x \mathcal{A}_u[\e]$, 
\begin{align}
\left\{\overline{f},\overline{g}\right\}_a :=  \overline{\frac{\delta  \overline{f}}{ \delta u(x)} \mathcal P_a \biggl(\frac{\delta \overline{g}}{ \delta u(x)}\biggr)} \, \in \mathcal{A}_u[\e] / \partial_x \mathcal{A}_u[\e],\quad a=1,2. \label{Poisson12}
\end{align}
The commutativity of KdV flows~\eqref{KdV} can then be viewed equivalently from either
\begin{align}
\label{Poisson1}
\left\{ \overline{h_m^{\sf cl}} \,,\, \overline{h_n^{\sf cl}} \right\}_1 &= 0 , \quad \forall\,m,n\geq 0 ,
\intertext{or} 
\label{Poisson2}
\left\{ \overline{h_{m-1}^{\sf cl}} \,,\,  \overline{h_{n-1}^{\sf cl}} \right\}_2 &= 0 , \quad \forall\,m,n\geq 0 .
\end{align}

Quantization of the KdV hierarchy with respect to the first\footnote{Quantization of the KdV hierarchy with respect to the second Poisson structure $\{,\}_2$ 
was considered by Bazhanov, Lukyanov and Zamalodchikov~\cite{Bazhanov1,Bazhanov2,Bazhanov3}. 
} 
Poisson structure $\{,\}_1$ was considered by Buryak and Rossi~\cite{BR2016a}, Dubrovin~\cite{D2016} and Eliashberg~\cite{E2007}.
 Let us now give a brief review. 
Consider the coordinate functional~$u(x)$, and write it using the Fourier modes representation:
\begin{equation}\label{fmodes}
u(x) = \sum_{k\in \mathbb{Z}} U_k {\rm e}^{\i kx}.
\end{equation}
Here and below, we will consider $x$ as a variable in $S^1$, and for any $f=f(u;u_x,u_{xx},\dots;\e)\in \mathcal{A}_u[\e]$, we  
regard $\overline{f}$ as $\int_{S^1}f(u;u_x,u_{xx},\dots;\e)\,\d x$ (hereafter $\int_{S^1}$ denotes $\frac 1{2\pi}\int_0^{2\pi}$, 
the operator extracting the zero Fourier mode).
In terms of the Fourier modes, the first Poisson structure in~\eqref{Poisson12} reads as follows:
\be
\label{PB}
\{U_k, U_l\}_1 = \i k \delta_{k+l,0} \,,\quad \forall\,k,l\in\mathbb{Z}.
\ee
The canonical quantizations with respect to $\{\cdot,\cdot\}_1$ for the space of local functionals 
$\mathcal{A}_u[\e]/\partial_x \mathcal{A}_u[\e]$ are defined as the following linear operators on~$\Lambda[[\e]]$:
\begin{align}
 \wh U_k &= \begin{cases} q_k, & k>0, \\ 
  U_0\, {\rm id}, & k=0, \\
  -\hbar k  \frac{\partial }{\partial q_{-k}}, & k < 0, \\ 
  \end{cases} \\
 \wh {\overline{a}} &= \, \, \, :  \int_{S^1} a(u(x);u_x(x),u_{xx}(x),\dots;\e)\, \d x : \, , \quad \forall\,a\in \mathcal{A}_u[\e], 
\end{align}
where $U_0$ and $\hbar$ are parameters and $:\,:$ denotes the normally ordered product. The normal ordering here means the following: we first move all $U_k$ with $k>0$ to the left of the $U_k$ with $k<0$ in each monomial of~$a$, and then we replace each $U_k$ by~$\wh U_k$.
In such a canonical quantization, the Poisson bracket \eqref{PB} is replaced 
by $\frac 1{\i\hbar}[\cdot,\cdot]$, where $[\cdot,\cdot]$ denotes the commutator;
indeed, 
$$
\left[\wh U_k, \wh U_l\right] = -\hbar k \delta_{k+l,0},\qquad k,l\in \mathbb{Z}.
$$
It turns out that the Poisson commutativity~\eqref{Poisson1} does not hold after the quantization, namely $\left[\widehat{\overline{h_m^{\sf cl}}},\widehat{\overline{h^{\sf cl}_n}}\right]$ do not vanish in general. 
The quantization problem, suggested by Dubrovin, is to find an
$\hbar$-deformation~$\overline{h_m}$ of~$\overline{h_m^{\sf cl}}$, such that 
\be\label{commquantumkdv}
 \left[\widehat{\overline{h_m}}\,,\, \widehat{\overline{h_n}} \right] =0, \quad \forall\,m,n\geq 0.
\ee
For the case that the dispersive parameter~$\e$ vanishes, the existence of such an $\hbar$-deformation was 
obtained by Eliashberg~\cite{E2007}, in the setting of Symplectic Field Theory~\cite{EGH2000}, via an explicit generating series (cf.~\eqref{Eliashbergformula} below); for an arbitrary~$\e$, the existence is obtained 
via a concrete construction by Buryak and Rossi~\cite{BR2016a} (cf.~\eqref{qkdvHamiltonians1}-\eqref{qkdvHamiltonians2} below). 
Let us denote by $h_m^{\sf DR}$ the densities of the quantum Hamiltonians constructed by Buryak and Rossi, which of course satisfy  
$\overline{h_m^{\sf DR}}|_{\hbar =0}=\overline{h_m^{\sf cl}}$.
Here we note that adding a total $x$-derivative in the density of a local functional does not affect the quantization of the functional.
Below, for simplify we denote Buryak--Rossi's quantum KdV Hamiltonians~$\widehat{\overline{h_m^{\sf DR}}}$ as~$H_m$. 

In this paper, we consider the eigenvalue problem for the operators~$H_m$. 
For the reader's convenience, we list below the first few $H_m$:
\begin{align}
& H_{-1} = U_0, \\
& H_0 =  \hbar\sum_{k>0}kq_k\frac{\pa}{\pa q_k}-\frac\hbar{24}+\frac {U_0^2}2, \\
& H_1 = \Delta + \hbar U_0\sum_{k\geq 1}kq_k\frac{\pa}{\pa q_k}  - \frac{\e^2\hbar}{12}\sum_{k\geq 1}k^3q_k\frac{\pa}{\pa q_k}
-\frac {\e^2\hbar}{2880} -\frac{\hbar U_0}{24} +\frac {U_0^3}6 ,
\end{align}
where 
\be\label{cutandjoin}
\Delta := \frac 12 \sum_{i,j\geq 1}\hbar(i+j)q_iq_j\frac{\pa}{\pa q_{i+j}}+\hbar^2ij q_{i+j}\frac{\pa^2}{\pa q_i\pa q_j}.
\ee
Here we note that $H_0$ is a grading operator, and that $\Delta$ can be recognized 
as the celebrated \emph{cut-and-join} operator (cf.~\cite{G1994} and Appendix~\ref{AppYJM}).
In general, write 
\be\label{genusexpansionHm}
H_m = \sum_{g\geq 0} \e^{2g} H_m^{[g]},
\ee
where $H_m^{[g]}$ are independent of $\epsilon$.
It is helpful to notice that 
for every $m\geq-1$, the operator $K_m:=H_m|_{U_0=0}/\sqrt{\hbar^{m+2}}$, after the rescaling 
\be\label{rescalingqt}
q \; \to \; q/\sqrt\hbar=:T,
\ee
depends only on the single parameter~$\sigma = -\frac{\e^2}{\sqrt\hbar}$ (cf.~Remark~\ref{remarkhomogeneous}). Here, $T=(T_1,T_2,T_3,\dots)$. 

The eigenvalue problem for $H_m$ with $\e=0$, i.e. for~$H^{[0]}_m$, was solved by Dubrovin~\cite{D2016}. 
To state his result, let us introduce some notations.  

Denote by $\Y$ the set of partitions $\ll$, namely $\ll=(\ll_1,\ll_2,\dots)\in\Y$ is a half-infinite non-increasing sequence of integers $\ll_1\geq\ll_2\geq \cdots$, called parts of $\ll$, with only finitely many non-zero parts; the number $\ell(\ll)$ of nonzero parts is called the length of the partition, and the number $|\ll|:=\sum_i\ll_i$ is called the weight of the partition. Let us denote $\Y_k$ the set of partitions $\ll$ having weight $|\lambda|=k$.
Let $s_\ll(T)$ denote the Schur polynomial associated to~$\ll$ which can be defined by 
\be
\label{JacobiTrudi}
s_\lambda(T)=\det (h_{\lambda_i-i+j}(T))_{1\leq i,j\leq \ell(\lambda)}.
\ee
Here $h_k(T),\,k\in \mathbb{Z}$ are the complete homogeneous symmetric polynomials defined via
\be
\sum_{k\geq 0} h_k(T) z^k := \exp\biggl(\sum_{k\geq 1} \frac {T_k}k z^k\biggr),\qquad h_k:=0,~{\rm if}~k<0.
\ee
For example, 
\be
s_{(2)}=\frac 12(T^2_1+T_2),\quad s_{(1,1)}=\frac 12(T_1^2-T_2),
\ee
\be
s_{(3)}=\frac 16(T_1^3+3T_1T_2+2T_3),\quad s_{(2,1)}=\frac 13(T_1^3-T_3),\quad s_{(1,1,1)}=\frac 16(T_1^3-3T_1T_2+2T_3).
\ee
The collection of $s_\ll(T)$ for all partitions form a basis of $\C\left[T_1,T_2,\dots\right]=:\widetilde{\Lambda}$ over~$\C$ (see e.g.~\cite{M2015}). 
We also recall that the standard inner product $\langle \,, \,\rangle:\widetilde{\Lambda}\times\widetilde{\Lambda}\to\C$ is 
defined by (see e.g.~\cite{M2015})
\be
\label{eq:schurorthogonal}
\left\langle s_\ll(T),s_\mu(T)\right\rangle=\delta_{\ll,\mu}.
\ee

Moreover, it is convenient to recall the following functions of partitions:
for any $j\geq 0$ and any partition $\ll\in\Y$, define \cite{BO2000,CMZ2018,D2016,Z2016}
\be
\label{P}
P_j(\ll)=\sum_{i=1}^{\ell(\ll)}\left[\left(\lambda_i-i+\frac{1}{2}\right)^j-\left(-i+\frac{1}{2}\right)^j\right]
\ee
and a differently normalized set of functions $Q_j$, for $j\geq 0$, related to the $P_j$'s by
\be
\label{Q}
Q_0(\ll)=1,\qquad Q_j(\ll)=\frac{P_{j-1}(\ll)}{(j-1)!}+\beta_j\quad\mbox{if }j\geq 1,
\ee
where $\beta_j$, for $j\geq 0$, are defined by the generating series
\be
\label{eq:beta}
\frac {z/2}{\sinh(z/2)} = \sum_{j\geq 0}\beta_jz^j=1-\frac{z^2}{24}+\frac{7z^4}{5760}-\frac{31z^6}{967680}+\cdots,
\ee
or, equivalently, $\beta_j=\left(\frac 1{2^{j-1}}-1\right)B_j/j!$, $B_j$ being the $j_{\rm th}$ Bernoulli numbers.
The functions $Q_j$, for $j\geq 0$, given in \eqref{Q} are the generators of the algebra of \emph{shifted symmetric polynomials} \cite{BO2000}. They are building blocks of the Bloch--Okounkov theorem and have a famous relation to quasimodular forms \cite{BO2000,CMZ2018,Z2016} which we hope to come back to in a later publication.

\begin{theorem}[Dubrovin~\cite{D2016}]
\label{thmdubrovin}
The following formulae hold true:
\be
H^{[0]}_m s_\ll\Bigl(\frac q{\sqrt\hbar}\Bigr)=E^{[0]}_{m}(\ll;\hbar,U_0) s_\ll\Bigl(\frac q{\sqrt\hbar}\Bigr), \quad m\geq-1
\ee
with the eigenvalues given explicitly by
\begin{align}
\nonumber
E^{[0]}_{m}(\ll;\hbar,U_0)&=
\sum_{j=0}^{m+2}\frac{\sqrt{\hbar^{j}}\beta_{j} U_0^{m+2-j}}{(m+2-j)!}+
\sqrt\hbar\sum_{i\geq 1}
\frac{\left(U_0+\sqrt\hbar\left(\ll_i+\frac 12-i\right)\right)^{m+1}-
\left(U_0+\sqrt\hbar\left(\frac 12-i\right)\right)^{m+1}}{(m+1)!}
\\
\label{dispersionlesseigenvalues}
&=\sum_{j=0}^{m+2}\frac{\sqrt{\hbar^{j}}U_0^{m+2-j}}{(m+2-j)!}Q_j(\ll),
\end{align}
where $Q_j(\ll)$ are defined in \eqref{Q}.
\end{theorem}

\smallskip

For a general $\e$, to study the eigenvalue problem for~$H_m$, let us make two observations. 
The first one is a symmetry property (cf.~Corollary~\ref{corollarysymmetry}) 
of the quantum KdV Hamiltonians,  
and the second one is that the spectrum of the quantum dispersionless KdV Hamiltonians is {\it non-degenerate} (cf.~Lemma \ref{nondegeneracylemma}).
Based on these observations we have the following proposition.

\begin{shaded}
\begin{proposition}
\label{prop:existencedeformedschur}
There exists a unique collection $\left(r_\ll(T;\sigma)\right)_{\ll\in\Y}$ of elements in the free module $\widetilde{\Lambda}\otimes_\C\C\llbracket\sigma\rrbracket$ over $\C\llbracket\sigma\rrbracket$ such that
\begin{align}
\label{1}
\left\langle r_\ll(T;\sigma),s_\ll(T)\right\rangle&=1,
\\
\label{2}
H_m r_\ll\Bigl(\frac q{\sqrt\hbar};-\frac{\e^2}{\sqrt\hbar}\Bigr)&=E_m\left(\ll;\e,\hbar,U_0\right)r_\ll\Bigl(\frac q{\sqrt\hbar};-\frac{\e^2}{\sqrt\hbar}\Bigr)
\end{align}
for some $E_m(\ll;\e,\hbar,U_0)\in\C\bigl(\sqrt\hbar,U_0\bigr)\llbracket\e^2\rrbracket$.
Here $H_m$ are Buryak--Rossi's quantum KdV Hamiltonians, cf. \eqref{qkdvHamiltonians1}-\eqref{qkdvHamiltonians2}.
Moreover, $\left(r_\ll(T;\sigma)\right)_{\ll\in\Y}$ is a basis of the free module $\widetilde{\Lambda}\otimes_\C\C\llbracket\sigma\rrbracket$, and 
\be
\label{3}
r_\ll(T;\sigma=0)=s_\ll(T), \qquad E_m(\ll;\e=0,\hbar)=E_m^{[0]}(\ll;\hbar),
\ee 
where $E_m^{[0]}$ are defined in~\eqref{dispersionlesseigenvalues}.
\end{proposition}
\end{shaded}

The proof is given in Section~\ref{sec:deformed}. 

We call $r_\ll(T;\sigma)$ in the above proposition the {\it deformed Schur polynomials}. 
For example, we have
\be\label{example2}
r_{(2)} =s_{(2)}+s_{(1,1)}\frac{4-\sqrt{16-\s^2}}\s =s_{(2)}+s_{(1,1)}\left(
-\frac{\s}{8}+\frac{\s^3}{512}-\frac{\s^5}{16384}+\frac{5 \s^7}{2097152}+\O(\s^{9})\right).
\ee
More examples can be found in Section~\ref{section32}. It should be noted that 
$\langle r_\lambda(T;\sigma), r_\mu(T;\sigma) \rangle =0$ for all $\lambda\neq\mu$ and that 
$r_{\ll'}(T;\s)=(-1)^{|\ll|}r_\ll(-T;-\s)$ (see Proposition~\ref{propositionconjugate}). Here $\ll'$ denotes 
the partition conjugate to~$\ll$. As an 
application of Proposition~\ref{prop:existencedeformedschur}, we also find 
certain vanishing identities for combinations of characters in the symmetric group (see Section~\ref{sectionvanishingid}).  

It is suggested by Dubrovin~\cite{Dubrovin} that to understand geometry of the spectrum and 
the common eigenvectors of~$H_m$ 
one needs to study the corresponding {\it spectral curves}. To define these curves, 
introduce a gradation on~$\widetilde{\Lambda}$ via the degree assignments $\deg T_k=k$, $k\ge1$. 
We have $\widetilde{\Lambda}=\bigoplus_{k\geq0}\widetilde{\Lambda}_k$, where 
elements in~$\widetilde{\Lambda}_k$ are homogeneous of degree~$k$ with respect to~$\deg$.
The commutativity~\eqref{commquantumkdv} with $n=0$ 
implies that each $K_m$ ($m\geq-1$) acts on the space $\widetilde{\L}_k$ for any $k\geq0$.
The weight~$k$ spectral curves for the quantum KdV hierarchy are then defined by 
\be
\label{spectralcurves}
\Sigma_{k,m}:=\left\lbrace(\s,\rho)\in\C^2:\ {\det}_{\widetilde{\L}_k}(K_m(\s)-\rho)=0\right\rbrace, \quad m\geq1.
\ee
We will discuss the geometric meanings of these curves after the proof of Proposition~\ref{prop:existencedeformedschur} (see Section~\ref{section32}). 
Dubrovin~\cite{Dubrovin} computes the geometric genus of the spectral curve $\Sigma_{k,1}$ up to~$k\leq 7$.
We continue his computation; below is a table of the geometric genus $g(\Sigma_{k,1})$.
\be
\begin{array}{l|cccccccccccc}
k & 0 & 1 & 2 & 3 & 4 & 5 & 6 & 7 & 8 & 9 & 10 &\cdots
\\ \hline
g(\Sigma_{k,1}) & 0 & 0 & 0 & 1 & 4 & 9 & 21 & 37 & 69 & 113 & 187 & \cdots
\end{array}
\ee
We have the following conjectural statements.

\begin{shaded}
\begin{conjecture} For any fixed $k\geq1$, 
the curves $\Sigma_{k,m}$ ($m\geq1$) are irreducible with geometric 
genus $g(\Sigma_{k,m})$ independent of~$m$ and given explicitly by the expression
 \be\label{gconj}
 g(\Sigma_{k,m})=(k-1)\left|\Y_k\right|+1-\sum_{\ll\in\Y_k}\ell(\ll),
 \ee 
where $\Y_k$ denotes the set of partitions of weight~$k$.
\end{conjecture}
\end{shaded}

It is straightforward to verify that the values of the right-hand side of~\eqref{gconj} 
coincide with those of sequence A238641 in The On-Line Encyclopedia of Integer Sequences\footnote{Available at \url{https://oeis.org/A238641}}, introduced by Kimberling. 
The conjectural identity~\eqref{gconj} was actually observed with the help of A238641. 

\smallskip

Recall that the KdV hierarchy is a particular Dubrovin--Zhang (DZ) hierarchy~\cite{DZ-norm}, that corresponds to the 
trivial rank-1 Cohomological Field Theory (CohFT). 
For the definition of a CohFT see~\cite{KM1994} or Section~\ref{secBR}.
For an arbitrary tautological CohFT (e.g. homogeneous CohFT, or Hodge case), following Dubrovin, one can 
consider the quantization of the first Poisson structure for the corresponding DZ hierarchy (\cite{BPS1,BPS2,DLYZ2016,DZ-norm}).
However, this Poisson structure is usually 
a deformation of $\eta^{\alpha\beta} \pa_x$, where $\eta^{\a\b}$ are certain constants. 
To make the quantization procedure as simple as for the KdV case, one could perform a Miura-type transform reducing 
the Hamiltonian operator for the DZ hierarchy to~$\eta^{\alpha\beta} \pa_x$ (cf.~\cite{DMS2005, DZ-norm,Getzler2002}) 
before doing the quantization. 
Considering the conjectural relationship between the DZ and the DR hierarchies associated to the 
same CohFT, it seems easier to directly quantize the DR hierarchy 
because the latter by definition possesses~$\eta^{\alpha\beta} \pa_x$ as the first Hamiltonian operator. 
Such a quantization is successfully constructed by Buryak and Rossi~\cite{BR2016a} with 
explicit and recursive formulae for the associated quantum DR Hamiltonians, which we will review in Section~\ref{secBR}. 
An analogous statement to Proposition~\ref{prop:existencedeformedschur} will also be given for the quantum DR hierarchy associated to a rank-1 CohFT.

\paragraph*{Organization of the paper.}
 
In Section~\ref{secBR}, we review the construction by Buryak and Rossi of the quantum DR hierarchy and prove 
a general symmetry property of the quantum DR Hamiltonians. 
In Section~\ref{sec:deformed}, we study the spectral problem of the quantum KdV hierarchy in details, in particular we prove Proposition~\ref{prop:existencedeformedschur}.
In Section~\ref{furtherremarks} we 
give further remarks including an analogous statement to Proposition~\ref{prop:existencedeformedschur} for an arbitrary rank-1 CohFT.
In Appendix~\ref{AppYJM} we give a representation of the quantum dispersionless KdV Hamiltonians 
in terms of multiplication operators in the class algebra of the symmetric group.

\paragraph*{Acknowledgements.}
We are grateful to Boris Dubrovin for his advice and helpful discussions.
We thank Konstantin Aleshkin, Marco Bertola, John Harnad, Paolo Rossi, Zuoqin Wang, Don Zagier and Youjin Zhang for valuable conversations.
This project has received funding from the 
National Key Research and Development Project ``Analysis and Geometry on Bundles" 
 SQ2020YFA070080, from NSFC No.~12061131014, and from 
the European Union's H2020 research 
and innovation programme under the Marie Sk\l owdoska-Curie grant No.~778010 {\em IPaDEGAN}. 
G.R. acknowledges support from the Fonds de la Recherche Scientifique-FNRS under EOS project O013018F. 
Part of our work was done at SISSA; we thank SISSA for excellent working conditions and financial support. 
G.R. also wishes to thank the School of Mathematical Sciences at University of Science and Technology of China in Hefei for warm 
hospitality where part of this work was carried out. 

\section{Buryak--Rossi's quantum DR Hamiltonians and their symmetry property}\label{secBR}

In this section, we review the construction of classical and quantum double ramification (DR for short) Hamiltonians~\cite{B2015, BR2016a, BR2016b} for a rank~$l$ Cohomological Field Theory (CohFT). We will restrict to the rank-1 case in the later sections.

Let $\Mgn$ be the moduli space of stable algebraic curves over $\C$ of genus~$g$ with~$n$ distinct marked points, and let 
$c_{g,n}\in\Hom(V^{\otimes n},H^{\sf even}(\Mgn,\C))$ be a CohFT~\cite{KM1994}. 
This means that we are given an $l$-dimensional complex vector space $V$, endowed with a non-degenerate symmetric two-form $\eta\in\mathrm{Sym}^2(V^*)$ and a distinguished element $\1\in V$, along with a collection of even cohomology classes $c_{g,n}(v_{1}\otimes\cdots\otimes v_{n})\in H^{\sf even}(\Mgn,\C)$, linearly depending on $v_1\otimes\cdots\otimes v_n\in V^{\otimes n}$, which satisfy the following axioms. To state them, denote by $e_1=\1,e_2,\dots,e_l$ a basis of~$V$, and let 
$\eta_{\a\b}:=\eta(e_\a,e_\b)$ and $\eta^{\a\b}$ the entries of~$\eta^{-1}$. Here and below, free Greek indices take the integer values $1,\dots,l$.
\begin{itemize}
\item $c_{g,n}$ is $\mathfrak S_n$-equivariant, with respect to the action of the symmetric group $\mathfrak S_n$ on $V^{\otimes n}$ permuting copies of $V$ and on $\Mgn$ permuting marked points.
\item Denoting $gl_1:\overline{\mathcal{M}}_{g-1,n+2}\to\Mgn$ the gluing morphism which identifies the last two marked points of a stable curve, we have
\be
gl^*_1 c_{g,n}(v_1\otimes\cdots\otimes v_n)=\eta^{\a\b}c_{g-1,n+2}(v_1\otimes\cdots\otimes v_n\otimes e_\a\otimes e_\b).
\ee
Here and in what follows, Einstein summation convention is assumed for contracted Greek indexes.
\item Denoting $gl_2:\overline{\mathcal{M}}_{g_1,n_1+1}\times\overline{\mathcal{M}}_{g_2,n_2+1}\to\Mgn$, $g=g_1+g_2$ and $n=n_1+n_2$, the gluing morphism which identifies the last marked points of a pair of stable curves, we have
\be
\qquad\qquad gl^*_2 c_{g,n}(v_1\otimes\cdots\otimes v_n)=\eta^{\a\b}c_{g_1,n_1+1}(v_1\otimes\cdots\otimes v_{n_1}\otimes e_\a)\otimes c_{g_2,n_2+1}(v_{n_1+1}\otimes\cdots\otimes v_{n}\otimes e_\b),
\ee
where in the right side we sum over the contracted indexes $\a,\b$.
\item Denoting $p:\overline{\mathcal{M}}_{g,n+1}\to\Mgn$ the map forgetting the last marked point of a stable curve, we have
\be
p^*c_{g,n}(v_1\otimes\cdots\otimes v_n)=c_{g,n+1}(v_1\otimes\cdots\otimes v_n\otimes e_1).
\ee
\item We have
\be
c_{0,3}(e_\a\otimes e_\b\otimes e_1)=\eta_{\a\b}.
\ee
\end{itemize}

We shall assume in the following that we can reduce $\eta$ to the form
\be
\label{eq:assumptioneta}
\eta_{\a\b}=\delta_{\alpha+\beta,l+1}
\ee
by a change of basis. This is always true for $l=1$, and follows by the assumption  $\eta(e_1,e_1)=0$ for $l\geq 2$, cf. \cite{Dubrovin2DTFTs}.
In what follows we shall also assume that the CohFT under consideration is \emph{tautological} (see \cite[Appendix B]{BR2016a}).

Given a rank~$l$ CohFT $c_{g,n}$, Buryak~\cite{B2015} associates to it an {\it integrable} hierarchy of Hamiltonian PDEs in~$l$ components $u^1,\dots,u^l$, called the \emph{DR hierarchy}.
Denote $u=(u^1,\dots,u^l)$, and 
 denote by $(\mathcal{A}_u,\partial_x)$  
the differential polynomial ring of~$u$, namely, an element of~$\mathcal{A}_u$ is a polynomial of 
$u_x^\a,u_{xx}^\a,u_{xxx}^\a,\dots$ whose coefficients are formal power series in~$u^\a$ with complex coefficients. The DR hierarchy has the form:
\be\label{DRcleqs}
\frac{\partial u^\alpha}{\partial t^{\beta,m}} =
\eta^{\a\rho}	\pa_x
 \Biggl(\frac{\delta \overline{h_{\beta,m}^{\sf DR,cl}}}{\delta u^\rho(x)} \Biggr), \quad m\geq 0,
\ee
where $h_{\beta,m}^{\sf DR,cl}$ are certain elements in $\mathcal{A}_u[[\e^2]]$, $m\geq 0$, defined below in \eqref{cHamiltoniandensities}. 
Here ``integrable" means that the flows commute pairwise:
\be\label{classint}
\frac{\pa}{\pa t^{\gamma,n}}\frac{\partial u^\alpha}{\partial t^{\beta,m}} =\frac{\pa}{\pa t^{\beta,m}}\frac{\partial u^\alpha}{\partial t^{\gamma,n}}.
\ee
Let us recall the construction of $h_{\beta,m}^{\sf DR,cl}$. Denote the Fourier expansion of~$u$ by
\begin{equation}
\label{fmodesDR}
u^\alpha(x) = \sum_{k\in \mathbb{Z}} U^\alpha_k {\rm e}^{\i kx}.
\end{equation}
For $m=-1$, they are defined by
\be
h_{\a,-1}^{\sf DR,cl}=\sum_{k\in\Z}\eta_{\a\b}U^\b_k{\rm e}^{\i kx},
\ee
and, for $m\geq 0$, by
\be
\label{cHamiltoniandensities}
h_{\a,m}^{\sf DR,cl}=\sum_{\begin{smallmatrix}g,n\geq 0\\ 2g-2+n\geq 0\\ k_1,\dots,k_n\in\mathbb{Z}
\end{smallmatrix}}
\frac{\left(-\epsilon^2\right)^g}{n!}\left(\int_{DR_g\big(-\sum\limits_{i=1}^n k_i,k_1,\dots,k_n\big)}\psi_1^m\ll_gc_{g,n+1}
\left(e_\a\otimes \bigotimes_{i=1}^n e_{\a_i}\right)\right)U_{k_1}^{\a_1}\cdots U_{k_n}^{\a_n}{\rm e}^{\i x\sum\limits_{i=1}^nk_i},
\ee
where $DR_g(k_0,\dots,k_n)\in H_{2(2g-2+n)}(\overline{\mathcal{M}}_{g,n+1},\C)$ is the DR cycle \cite{BSSZ2015,FP2005,JPPZ2017},
$\psi_1\in H^2(\Mgn,\C)$ is the $\psi$-class, and $\ll_g\in H^{2g}(\Mgn,\C)$ is the $g$-th Chern class of the Hodge bundle. 
It is proved in~\cite{B2015} that one can view the above-defined $h_{\a,m}^{\sf DR,cl}$ as elements in $\mathcal{A}_u[[\e^2]]$;
for example, $h_{\a,-1}^{\sf DR,cl}$ is identified with the differential polynomial $\eta_{\a\b}u^\b$.
The Hamiltonian operator $\eta^{\alpha\beta}\pa_x$ defines a Poisson structure on the space of local 
functionals~$\mathcal{A}_u[[\e^2]] / \partial_x \mathcal{A}_u[[\e^2]]$ in the following way: 
for any $ \overline{f}, \overline{g}\in \mathcal{A}_u[[\e^2]] / \partial_x \mathcal{A}_u[[\e^2]]$, 
\begin{align}\label{PoissonDR}
\left\{ \overline{f},\overline{g}\right\} :=  \overline{\frac{\delta \overline{f}}{ \delta u^\alpha(x)} 
\eta^{\alpha\beta}\pa_x \biggl(\frac{\delta \overline{g}}{ \delta u^\beta(x)}\biggr)} \, \in \mathcal{A}_u[[\e^2]] / \partial_x \mathcal{A}_u[[\e^2]].
\end{align}
The commutativity~\eqref{classint} of the DR flows~\eqref{DRcleqs} then can also be interpreted as
\begin{equation}\label{PoissonDRcomm}
\left\{ \overline{h_{\alpha,m}^{\sf DR,cl}} \,,\, \overline{h_{\beta,n}^{\sf DR,cl}} \right\} = 0 , \quad \forall\,m,n\geq 0.
\end{equation}
In terms of the Fourier modes \eqref{fmodesDR} the Poisson structure \eqref{PoissonDR} reads
\be
\bigl\{U_k^\alpha,U_l^\beta\bigr\} = \i k \eta^{\alpha\beta} \delta_{k+l,0},\quad \forall\,k,l\in\mathbb{Z}.
\ee
Similarly to the rank-1 case considered in the Introduction, let us recall that the canonical quantization with respect to $\{\cdot,\cdot\}$ for the space of local functionals $\mathcal{A}_u[\e]/\partial_x \mathcal{A}_u[\e]$ 
is defined as follows:
\begin{align}
 \wh U_k^\a &= \begin{cases} q^{\a,k}, & k>0, \\ 
  U^{\a}_0\, {\rm id}, & k=0, \\
  -\hbar k \eta^{\alpha\beta} \frac{\partial }{\partial q^{\beta,-k}}, & k < 0, \\ 
  \end{cases} \\
  \label{quanta}
 \wh {\overline{a}} &= \,  :  \int_{S^1} a(u(x);u_x(x),u_{xx}(x),\dots;\e)\, \d x : \, , \quad \forall\,a\in \mathcal{A}_u[\e], 
\end{align}
where $:\,:$ denotes the normally ordered product. As explained in the Introduction, this means that we first move all $U_k^\a$ 
with $k>0$ to the left of the $U_k^\a$ with $k<0$ in each monomial in $a$, 
and then we replace each $U_k^\a$ by the corresponding operator $\wh U_k^\a$.
The $\wh U_k^\a$ with $k\in \mathbb{Z}$ and the $\wh{\overline{a}}$ with $a\in \mathcal{A}_u[\e^2]$ 
are considered as operators on $\L$ depending on parameters $\hbar,\e$, i.e. formally as elements of $({\rm End}\,\L)\otimes\C[\e,\hbar]$.
In such a canonical quantization, the Poisson bracket is replaced by $\frac 1{\i\hbar}[\cdot,\cdot]$; namely, we have
\be
\Bigl[\wh U_k^\a, \wh U_j^\b\Bigr] = - \hbar k \eta^{\a\b} \delta_{k+j,0},\qquad k,j\in \mathbb{Z}.
\ee
It turns out that the Poisson commutativity~\eqref{PoissonDRcomm} does not hold after the quantization, namely, 
\be
\left[\widehat{\overline{h_{\alpha,m}^{\sf DR,cl}}}\,, \,\widehat{\overline{h^{\sf DR,cl}_{\beta,n}}}\right]
\ee 
in general do not vanish. 
One needs to deform~$h_{\a,m}^{\sf DR,cl}$ such that the quantizations of the deformed Hamiltonians commute.
A successful and explicit deformation is obtained by Buryak and Rossi using the DR cycles in the moduli spaces of curves \cite{BR2016a}. Indeed,  
introduce a family of elements $h_{\a,m}\in \mathcal{A}_u[[\e^2,\hbar]]$, $m\geq -1$ by
\be
h_{\a,-1}^{\sf DR}=\sum_{k\in\Z}\eta_{\a\b}U^\b_k{\rm e}^{\i kx}
\ee
and, for $m\geq 0$,
\be
\label{qHamiltoniandensities}
h_{\a,m}^{\sf DR}\!=\!\!\!\!\sum_{\begin{smallmatrix}g,n\geq 0\\ 2g-2+n\geq 0\\ k_1,\dots,k_n\in\mathbb{Z}
\end{smallmatrix}}\!\!
\frac{\hbar^g}{n!}\left(\int_{DR_g\big(-\sum\limits_{i=1}^n k_i,k_1,\dots,k_n\big)}\psi_1^m\L\left(-\frac{\epsilon^2}{\hbar}\right)c_{g,n+1}\left(e_\a\otimes \bigotimes_{i=1}^ne_{\a_i}\right)\right) U_{k_1}^{\a_1}\cdots U_{k_n}^{\a_n} {\rm e}^{\i x\sum\limits_{i=1}^nk_i},
\ee
where $\L(\xi)=1+\ll_1\xi+\cdots +\ll_g\xi^g$ is the Chern polynomial of the Hodge bundle, $\ll_i\in H^{2i}(\Mgn,\C)$.
Again, we can and should view $h_{\a,m}$ as elements in~$\mathcal{A}_u[\e,\hbar]$ via~\eqref{fmodesDR};
for example, $h_{\a,-1}=\eta_{\a\b} u^\b$. For more details about this correspondence with differential polynomials see \cite{B2015,BR2016a}.
Note also that $\left.h_{\a,m}^{\rm DR}\right|_{\hbar=0}=h_{\a,m}^{\sf DR,cl}$ for all $m\geq -1$.

It is shown by Buryak and Rossi \cite{BR2016a} that, denoting for convenience $H_{\a,m}=\widehat{\overline{h^{\sf DR}_{\alpha,m}}}$,
\be
\left[H_{\a,m},H_{\beta,n}\right] =0, \quad m,n\geq -1.
\ee
(We note that obtaining $H_{\a,m}$ is more direct from~\eqref{qHamiltoniandensities}.)
In particular $H_{1,0}$ has the expression \cite[Lemma 3.6]{BR2016a}
\be
\label{degreegeneral}
H_{1,0}=\frac 12\eta_{\a\b}U_0^\a U_0^\b+\hbar\sum_{\a=1}^d\sum_{k>0}kq^{\a,k}\frac{\pa}{\pa q^{\a,k}}+{\rm const},
\ee
which is, up to constants, the grading operator with respect to the degree $\deg q^{\a,k}:=k$. 
From this we know that all $H_{\a,m}$  preserve such degree.

Although one can obtain $H_{\alpha,m}$ from~\eqref{qHamiltoniandensities} directly, in practice it is actually
more effective to compute them via the following recursion, found and proved by Buryak and Rossi \cite{BR2016a}:
\begin{align}
\label{BRrec1}
\frac{\pa}{\pa x}\left(D-1\right)h_{\a,m+1}^{\sf DR}&=\frac 1\hbar [h_{\a,m}^{\sf DR},H_{1,1}],\qquad D:=\epsilon\frac{\pa}{\pa\epsilon}+2\hbar\frac{\pa}{\pa\hbar}+\sum_{\a=1}^d\sum_{k\geq 0}u_{kx}^\a\frac{\pa}{\pa u_{kx}^\a},
\\
\label{BRrec2}
\frac{\pa h_{\a,m+1}^{\sf DR}}{\pa u^1}&=h_{\a,m}^{\sf DR}.
\end{align}

We now discuss a symmetry property of the quantum DR Hamiltonians.
The space $\L$ has a basis given by the monomials
\be
\label{basis}
q^{\a_1,k_1}\cdots q^{\a_n,k_n},
\ee
with $k_i>0$ and $\a_i=1,\dots,l$. Introduce a sesquilinear form on~$\L$, with values in $\C[\hbar]$ defined on the basis elements \eqref{basis} by
\begin{align}
\nonumber
\left\langle q^{\a_1,k_1}\cdots q^{\a_n,k_n},q^{\b_1,j_1}\cdots q^{\b_n,j_n}
\right\rangle
&=\hbar^n k_1\cdots k_n\sum_{\s\in\mathfrak S_n}\prod_{i=1}^n\delta_{k_i,j_{\s(i)}}\eta^{\a_i\b_{\s(i)}}
\\
\label{defscalarproduct}
&=\hbar^n k_1\cdots k_n\sum_{\s\in\mathfrak S_n}\prod_{i=1}^n\delta_{k_i,j_{\s(i)}}\delta_{\alpha_i+\beta_{\sigma(i)},l+1},
\end{align}
cf.~\eqref{eq:assumptioneta}, setting $\left\langle q^{\a_1,k_1}\cdots q^{\a_n,k_n},q^{\b_1,j_1}\cdots q^{\b_{m},j_{m}}\right\rangle=0$ whenever $n\not=m$.
We shall denote by the same symbol $\langle,\rangle$ the natural sesquilinear extension to $\L\otimes_\C\C[[\hbar,\e]]$.

\begin{remark}
For rank-1 CohFTs, the sesquilinear form~\eqref{defscalarproduct} reduces, up to a rescaling in~$\hbar$, to the standard inner product \eqref{eq:schurorthogonal} on the space of symmetric polynomials~\cite{M2015} (cf.~also Appendix~\ref{AppYJM}).
\end{remark}

\begin{lemma}
\label{lemmaadjoint}
For any $k\in\mathbb{Z}$, the operators $\wh U_k^\a$ and $\wh U_{-k}^\a$
satisfy the following relation:
\be
\label{symmetrylemma}
\left\langle \wh U_k^\a f,g\right\rangle=\left\langle f,\wh U_{-k}^\a g \right\rangle,\qquad f,g\in\L.
\ee
\end{lemma}
\begin{proof}
Note that the statement is trivial for $k=0$ and it is symmetric with respect to $k\mapsto -k$. 
So we only need to consider the case that $k>0$. Take $k=k_1>0$. 
To show the statement it suffices to verify~\eqref{symmetrylemma} when $f,g$ are the monomials~\eqref{basis}. 
The left-hand side of~\eqref{symmetrylemma} can be nonzero only for
\begin{align}
\nonumber
\left\langle \wh U_{k_1}^{\a_1} q_{k_2}^{\a_2}\cdots q_{k_{n+1}}^{\a_{n+1}},q_{j_1}^{\b_1}\cdots q_{j_{n+1}}^{\b_{n+1}}\right\rangle&=
\left\langle q_{k_1}^{\a_1} \cdots q_{k_{n+1}}^{\a_{n+1}},q_{j_1}^{\b_1}\cdots q_{j_{n+1}}^{\b_{n+1}}\right\rangle
\\
\label{firstside}
&=\hbar^{n+1}k_1\cdots k_{n+1}\sum_{\s\in\mathfrak S_{n+1}}\prod_{i=1}^{n+1}\delta_{k_i,j_{\s(i)}}\eta^{\a_i\b_{\s(i)}}.
\end{align}
On the other hand, the right-hand side of~\eqref{symmetrylemma} is nonzero only for
\begin{align}
\nonumber
&\left\langle q^{\a_2,k_2}\cdots q^{\a_{n+1},k_{n+1}},\wh U_{-k_1}^{\a_1}q^{\b_1,j_1}\cdots q^{\b_{n+1},j_{n+1}}\right\rangle
\\ \nonumber
&\qquad
=\hbar k_1\eta^{\a_1\beta}\left\langle q^{\a_2,k_2}\cdots q^{\a_{n+1},k_{n+1}},\frac{\pa}{\pa q^{\b,k_1}}\left(q^{\b_1,j_1}\cdots q^{\b_{n+1},j_{n+1}}\right)\right\rangle
\\
&\qquad
=\sum_{i=1}^{n+1}\hbar k_1\eta^{\a_1\beta_i}\delta_{k_1,j_i}\left\langle q^{\a_2,k_2}\cdots q^{\a_{n+1},k_{n+1}},q^{\b_1,j_1}\cdots \mathop{q^{\b_i,j_i}}^{\frown}\cdots q^{\b_{n+1},j_{n+1}}\right\rangle
\end{align}
where $\frown$ denotes omission of the corresponding term. By \eqref{defscalarproduct} we rewrite the last expression as
\be
\hbar^{n+1}k_1\cdots k_{n+1}\sum_{i=1}^{n+1}\eta^{\a_1\beta_i}\delta_{k_1,j_i}\sum_{\pi}\delta_{k_2,j_{\pi(2)}}\eta^{\a_2\b_{\pi(2)}}\cdots\delta_{k_{n+1},j_{\pi(n+1)}}\eta^{\a_{n+1}\b_{\pi(n+1)}}
\ee
where the internal sum ranges over all bijections $\pi:\{2,\dots,n+1\}\to\{1,\dots,n+1\}\setminus\{i\}$. Changing summation indexes $(i,\pi)\to\rho\in\mathfrak S_{n+1}$ where $\rho$ is defined by $\rho(1)=i,\rho(2)=\pi(2),\dots,\rho(n+1)=\pi(n+1)$ we recognize that the last expression is the same as \eqref{firstside}. The proof is complete.
\end{proof}

Denote by~$\mathcal A_u^\R$ the subalgebra of $\mathcal A_u$ consisting of differential polynomials with real coefficients.

\begin{lemma} 
\label{lemmasymmetry}
For any real differential polynomial $a\in\mathcal A_u^{\R}$, the operator~$A=\wh{\overline a}$ (cf. \eqref{quanta}) is (formally) self-adjoint with respect to~\eqref{defscalarproduct}:
\be
\left\langle Af,g\right\rangle
=
\left\langle f,Ag\right\rangle
,\qquad f,g\in\L.
\ee
\end{lemma}

\begin{proof}
It is enough to show the lemma when $a$ is a monomial of the form $a=\pa_x^{j_1}u^{\a_1}\cdots\pa_x^{j_n}u^{\a_n}$, for some $n>0$, $j_1,\dots,j_n\geq 0$, $\a_1,\dots,\a_n\in\{1,\dots,d\}$. In such case $A=\wh{\overline a}$ has the form
\be
\label{form}
A
=\sum_{\begin{smallmatrix}k_1,\dots,k_n\in\Z \\ k_1+\cdots+k_n=0\end{smallmatrix}} \left(\i k_1\right)^{j_1}\cdots \left(\i k_n\right)^{j_n}\mathop{:}U_{k_1}^{\a_1}\cdots U_{k_n}^{\a_n}\mathop{:}
\ee
and so we can explicitly compute the adjoint operator $A^*$, satisfying $\left\langle Af,g\right\rangle=\left\langle f,A^*g\right\rangle$ using Lemma~\ref{lemmaadjoint}. This yields
\be
A^*
=\sum_{\begin{smallmatrix}k_1,\dots,k_n\in\Z \\ k_1+\cdots+k_n=0\end{smallmatrix}} \left(-\i k_1\right)^{j_1}\cdots \left(-\i k_n\right)^{j_n}\mathop{:}U_{-k_n}^{\a_n}\cdots U_{-k_1}^{\a_1}\mathop{:}
\ee
and relabeling the indices of the sum by $k_i\mapsto -k_{i}$ we conclude that $A=A^*$.
\end{proof}

\begin{corollary}
\label{corollarysymmetry}
Assuming that the quantum DR Hamiltonians $H_{\a,m}=\wh{\overline{h_{\a,m}^{\sf DR}}}$ (cf.~\eqref{qHamiltoniandensities} for the definition of~$h_{\a,m}^{\sf DR}$) have real coefficients and assuming $U_0\in\R$, then $H_{\a,m}$ are (formally) self-adjoint with respect to~\eqref{defscalarproduct}.
\end{corollary}

\begin{proof}
The statement follows directly from Lemma~\ref{lemmasymmetry}.
\end{proof}

It follows from Corollary~\ref{corollarysymmetry} that, under the same reality assumptions, the Buryak--Rossi's quantum Hamiltonians are {\it diagonalizable}.

Let us now consider the example of the trivial CohFT $\bigl(V=\C$, $\eta=1$, $c_{g,n}=1\bigr)$. 
In this case, Buryak--Rossi's quantum Hamiltonian densities read as follows:
\begin{align}
\label{qkdvHamiltonians1}
h_{-1}^{\sf DR}&=u(x),
\\
\label{qkdvHamiltonians2}
h_{m}^{\sf DR}&=
\sum_{\begin{smallmatrix}g\geq 0,\ n\geq 0 \\ 2g-2+n\geq 0\end{smallmatrix}}\sum_{\begin{smallmatrix}k_1,\dots,k_n\in\Z \end{smallmatrix}}\sum_{\ell=0}^g\frac{(-\epsilon^2)^\ell\hbar^{g-\ell}}{n!}\left(\int_{DR_g\big(-\sum\limits_{i=1}^n k_i,k_1,\dots,k_n\big)}\psi_1^m\lambda_\ell\right)U_{k_1}\cdots U_{k_n},\quad m\geq 0.
\end{align}
Here we omit the index~$\a$, as we now have~$l=1$.
The corresponding quantum Hamiltonians~$H_m=\wh{\overline{h_m^{\sf DR}}}$ form Buryak--Rossi's quantum KdV hierarchy. 
Indeed, by comparing~\eqref{qkdvHamiltonians1}--\eqref{qkdvHamiltonians2} with~\eqref{cHamiltoniandensities} (taking $c_{g,n}\equiv1$) one sees that the classical $\hbar\to0$ limit of $\overline{h_m^{\sf DR}}$ gives the classical DR Hamiltonians associated to the trivial CohFT, which coincides with the classical KdV Hamiltonians~\cite{B2015}.
Recalling the expansion \eqref{genusexpansionHm} and using the dimensional constraint it can be proved that $H_m^{[i]}=0$ for $i>m+2$;
we observe that also $H_m^{[m+1]}=H_m^{[m+2]}=0$.

\begin{remark}\label{remarkhomogeneous}
Noting that the intersection number on the DR cycle in \eqref{qkdvHamiltonians2} vanishes unless $m+\ell=2g-2+n$, we find that $h_{m}$ is homogeneous of degree $m+2$ with respect to the grading
\be
\label{quantumKodama}
\wt\deg\,u=1,\qquad\wt{\deg}\,\e=\frac 12,\qquad\wt{\deg}\,\hbar=2.
\ee
\end{remark}

Denote by 
\be
\mathcal{H}^{[0]}(z):=1+\sum_{m\geq -1}z^{m+2}H_m^{[0]}
\ee
the generating series of the quantum \emph{dispersionless} KdV Hamiltonians, where $H_m^{[0]}$ are defined in~\eqref{genusexpansionHm}. 
The following formula is given by Buryak and Rossi~\cite{BR2016a}:
\be
\label{Eliashbergformula}
\mathcal{H}^{[0]}(z)=\frac 1{S(\sqrt\hbar z)} :\int {\rm e}^{z S\left(\i \sqrt\hbar z \frac{\pa}{\pa x}\right)u}\d x :\,,
\ee
where the series $S(z)$ is defined as
\be
\label{S}
S(z)=\frac{\sinh(z/2)}{z/2}=1+\frac{z^2}{24}+\frac{z^4}{1920}+\frac{z^6}{322560}+\cdots.
\ee
By comparing with Eliashberg's formula \cite{E2007}, 
one observes immediately that Buryak--Rossi's quantum \emph{dispersionless} KdV Hamiltonians coincide with 
Eliashberg's ones.

\section{The spectral problem of the quantum KdV hierarchy}\label{sec:deformed}

In this section we study the spectral problem of the quantum KdV hierarchy and give applications.

\subsection{Non-degeneracy property}

The following lemma plays a crucial role in the proof of Prop. \ref{prop:existencedeformedschur}.

\begin{lemma}\label{nondegeneracylemma}
For any $w\geq 0$ there exists $m=m(w)\geq 0$ that $E_m^{[0]}(\ll)\not =E_m^{[0]}(\mu)$ for all $\ll\not=\mu$, $|\ll|=w=|\mu|$. Here $E_m^{[0]}(\ll)$ are the dispersionless eigenvalues, given in \eqref{dispersionlesseigenvalues}. 
\end{lemma}

\begin{proof}
We first claim that $P_j(\ll)=P_j(\mu)$ for all $j\geq 0$ if and only if $\ll=\mu$, where $P_j$ are defined in \eqref{P}. 
To prove this, we consider the generating series
\be
\sum_{j\geq 1}P_j(\ll)\frac{z^j}{j!}=\sum_{i\geq 1}\left[{\rm e}^{z\left(\ll_i-i+\frac 12\right)}-{\rm e}^{z\left(-i+\frac 12\right)}\right]=\sum_{i\geq 1}{\rm e}^{z\left(\ll_i-i+\frac 12\right)}-\frac 1{2\sinh(z/2)},
\ee
provided that the complex variable $z$ satisfies $\Re z>0$. It follows that if $P_j(\ll)=P_j(\mu)$ for all $j\geq 0$ we have
\be
\sum_{i\geq 1}{\rm e}^{z\left(\ll_i-i\right)}=\sum_{i\geq 1}{\rm e}^{z\left(\mu_i-i\right)}
\ee
which implies $\ll_i=\mu_i$ for all $i\geq 1$, and so the partitions $\ll$ and $\mu$ must coincide, as claimed.
It follows that for any pair of partitions $\ll\not=\mu$, with $|\ll|=|\mu|$, there exists $j=j(\ll,\mu)$ such that
$Q_j(\ll)\not=Q_j(\mu)$, where $Q_j$ are defined in \eqref{P}. By looking at \eqref{dispersionlesseigenvalues}, namely
\be
E_m(\ll;\hbar,U_0)=\sum_{j=0}^{m+2}\frac{\sqrt{\hbar^{j}}U_0^{m+2-j}}{(m+2-j)!}Q_j(\ll),
\ee
we conclude that as soon as $m>j(\ll,\mu)$ for all $\ll,\mu$ of the same weight $w=|\ll|=|\mu|$, some coefficient of $U_0$ in $E^{[0]}_{m}(\ll;\hbar,U_0)$ is different from the corresponding coefficient in $E^{[0]}_{m}(\mu;\hbar,U_0)$.
\end{proof}

\subsection{Proof of Proposition~\ref{prop:existencedeformedschur}}\label{section32}

Let us introduce
\be
\label{scaling}
H_m=\sqrt{\hbar^{m+2}}K_m,\qquad U_k=\sqrt{\hbar} V_k,\qquad \sigma:=-\frac{\epsilon^2}{\sqrt\hbar},
\ee
so that \eqref{qkdvHamiltonians1} and \eqref{qkdvHamiltonians2} read as
\begin{align}
K_{-1}&=V_0,
\\
\label{scaling2}
K_m&=\sum_{\begin{smallmatrix}g\geq 0,\ n\geq 0 \\ 2g-2+n\geq 0\end{smallmatrix}}\sum_{\begin{smallmatrix}k_1,\dots,k_n\in\Z \\ k_1+\cdots+k_n=0\end{smallmatrix}}\frac{\sigma^{2g-2+n-m}}{n!}\left(\int_{DR_g(0,k_1,\dots,k_n)}\psi_1^m\lambda_{2g-2+n-m}\right): \wh{V}_{k_1}\cdots \wh{V}_{k_n}: \quad (m\geq 0),
\end{align}
where we have used the dimensional constraint $m+\ell=2g-2+n$ in \eqref{qkdvHamiltonians2} (cf. Remark \ref{remarkhomogeneous}).
Let us also introduce
\be
q_k=\sqrt{\hbar} T_k,\quad U_0=\sqrt{\hbar}V_0.
\ee
Then the operators $\wh V_k$ act on $\C[T_1,T_2,\dots]$. Explicitly,
\be
\wh V_k:f(T)\mapsto\begin{cases} T_k f(T), & k>0, \\ V_0 f(T), & k=0,\\  -k\frac{\pa f(T)}{\pa T_{-k}}, & k<0.\end{cases}
\ee

In particular $H^{[0]}_m=\sqrt{\hbar^{m+2}}K_m^{[0]}$, where $K_m^{[0]}:=K_m|_{\s=0}$; 
therefore Dubrovin's result (cf.~Theorem~\ref{thmdubrovin}) states that
\be
K_m^{[0]}s_\ll(T)=F^{[0]}_m(\ll;V_0)s_\ll(T),
\ee
where $s_\ll(T)$ are the Schur polynomials \eqref{JacobiTrudi}, and 
\begin{align}
F^{[0]}_m(\ll;V_0)&=\frac 1{\sqrt{\hbar^{m+2}}}E_m^{[0]}(\ll;\hbar,\sqrt\hbar V_0) 
=\sum_{j=0}^{m+2}\frac{V_0^{m+2-j}}{(m+2-j)!}Q_j(\ll), \label{dispersionlesseigenvaluesscaled}
\end{align}
which is also independent of~$\hbar$.
The property \eqref{2} in the statement of Proposition \ref{prop:existencedeformedschur} is then equivalent to
\be
\label{22}
K_m r_\ll\left(T;\sigma\right)=F_m\left(\ll;\sigma,V_0\right)r_\ll\left(T;\sigma\right),
\ee
for some $F_m(\ll;\sigma,V_0)\in\C\left(V_0\right)\llbracket\s\rrbracket$; 
the relation with $E_m(\ll;\hbar,\e,U_0)$ in the statement of Proposition~\ref{prop:existencedeformedschur} is
\be
E_m(\ll;\hbar,\e,U_0)=\sqrt{\hbar^{m+2}}F_m\biggl(\ll;\s=-\frac{\epsilon^2}{\sqrt\hbar},V_0=\frac{U_0}{\sqrt\hbar}\biggr)
\ee

Therefore, let us first prove that \eqref{1} and \eqref{2} (the latter equivalently rewritten as~\eqref{22}) imply~\eqref{3}. Indeed, looking at \eqref{22} for $\s=0$ we obtain that $r_\ll(T;\s=0)$ is a basis on which $H_m^{[0]}$ are diagonal, and it follows from Theorem~\ref{thmdubrovin} 
that  $r_\ll(T;\s=0)$ must be a linear combination of the Schur polynomials $s_\mu(T)$ for which $E_m^{[0]}(\ll)=E_m^{[0]}(\mu)$ for all $m\geq -1$; Lemma~\ref{nondegeneracylemma} implies then that $r_\ll(T;\s=0)=c \, s_\ll(T)$, for some $c\in\C$, and $E_m(\ll;\s=0)=E_m^{[0]}(\ll)$. Finally $c=1$ due to \eqref{1}.

Let us proceed and prove the existence and the uniqueness of~$r_\ll(T;\s)$. 
Denote~$r_\ll(T;\s)=\sum_{k\geq 0}r_\ll^{[k]}(T)\s^k$ and $F_m(\ll;\s)=\sum_{k\geq 0} F_m^{[k]}(\ll)\s^k$; we will show existence and uniqueness of the coefficients $r_\ll^{[k]}(T)$ and $F_m^{[k]}$ for $k\geq 0$, such that \eqref{1} and \eqref{2} hold true. 

The case $k=0$, namely that $r_\ll^{[0]}(T)=s_\ll(T)$ and that $F_m^{[0]}(\ll)$ is given by \eqref{dispersionlesseigenvaluesscaled}, follows from Dubrovin's Theorem \ref{thmdubrovin}, as explained above.
Next, consider \eqref{22} at order $\s^k$ for some $k\geq 1$ and take the inner product, using \eqref{eq:schurorthogonal}, 
with $s_\mu(T)$ (for $|\mu|=|\ll|$)
\be
\sum_{j=0}^k\left\langle s_\mu (T),K_m^{[j]}r_\ll^{[k-j]}(T)\right\rangle=\sum_{j=0}^k F_m^{[j]}(\ll)\left\langle s_\mu(T),r_{\ll}^{[k-j]}(T)\right\rangle.
\ee
Denote $\left\langle r_\ll^{[\ell]}(T),s_\ll(T)\right\rangle=r_{\ll\mu}^{[\ell]}$; the normalization \eqref{1} implies that $r_{\ll\ll}^{[\ell]}=\delta_{\ell,0}$. Exploiting the symmetry property of the $H_m$ (Corollary \ref{corollarysymmetry}) we infer
\be 
\sum_{j=0}^k\left\langle K_m^{[j]} s_\mu (T),r_\ll^{[k-j]}(T)\right\rangle=\sum_{j=0}^k F_m^{[j]}(\ll)r_{\ll\mu}^{[k-j]},
\ee
and by separating terms corresponding to $j=0$ in the sums, we obtain
\be
\left(F_m^{[0]}(\ll)-F_m^{[0]}(\mu)\right)r_{\ll\mu}^{[k]}=\sum_{j=1}^k\left\langle K_m^{[j]}s_\mu (T),r_\ll^{[k-j]}(T)\right\rangle-\sum_{j=1}^k F_m^{[j]}(\ll)r_{\ll\mu}^{[k-j]}.
\ee
Write then the last relation in the case $\mu=\ll$
\be
\label{orderbyorder1}
F_m^{[k]}(\ll)=\sum_{j=1}^k\sum_{\nu\in\Y_{|\ll|}}\left\langle K_m^{[j]}s_\ll (T),s_\nu(T)\right\rangle r_{\ll\nu}^{[k-j]},
\ee
and in the case $\mu\not=\ll$
\be
\label{orderbyorder2}
\left(F_m^{[0]}(\ll)-F_m^{[0]}(\mu)\right)r_{\ll\mu}^{[k]}=\sum_{j=1}^k\left\langle K_m^{[j]}s_\mu (T),r_\ll^{[k-j]}(T)\right\rangle-\sum_{j=1}^{k-1} F_m^{[j]}(\ll)r_{\ll\mu}^{[k-j]}.
\ee
For any $w\geq 0$, let us first consider $m_*\geq -1$ such that $F_{m_*}^{[0]}(\lambda)\not=F_{m_*}^{[0]}(\mu)$ whenever $\lambda,\mu$ are distinct partitions of $w$ (cf. Lemma \ref{nondegeneracylemma}).
Since we already know $F_{m_*}^{[0]}$ and $r_{\lambda\mu}^{[0]}$, we can use \eqref{orderbyorder1} by induction on $k$ to obtain $F_{m_*}^{[k]}(\lambda)$ for all $k\geq 0$ and all $\lambda$ partitions of $w$ and to obtain $r_{\ll\mu}^{[k]}$ for all $k\geq 0$ and all distinct partitions $\ll\not=\mu$ of $w$.
This proves uniqueness of $r_\lambda(T;\s)$ and $F_{m_*}(\lambda;\s)$; moreover, by construction, we conclude that \eqref{22} with $m=m_*$ holds true for all $\lambda$ partitions of $w$.
Next, for $\lambda\not=\mu$ being any distinct partitions of~$w$,  we have, exploiting symmetry of the Hamiltonians,
$$
F_{m_*}(\mu)\left\langle r_\lambda,r_\mu\right\rangle
=\left\langle r_\lambda,K_{m_*}r_\mu\right\rangle
=\left\langle K_{m_*} r_\lambda,r_\mu\right\rangle
=F_{m_*}(\lambda)\left\langle r_\lambda,r_\mu\right\rangle,$$
which proves, thanks to Lemma \ref{nondegeneracylemma}, that $\left\langle r_\lambda,r_\mu\right\rangle$ vanishes. Thus $r_\lambda(T;\s)$ form a basis of the free module $\widetilde{\Lambda}\otimes_\C\C\llbracket\sigma\rrbracket$.
Finally, for any $m$, by commutativity and symmetry of the Hamiltonians,
$$
F_{m_*}(\mu)\left\langle K_m r_\lambda,r_\mu\right\rangle
=\left\langle K_m r_\lambda,K_{m_*}r_\mu\right\rangle
=\left\langle K_{m_*} K_m r_\lambda,r_\mu\right\rangle
=\left\langle  K_m K_{m_*}r_\lambda,r_\mu\right\rangle
=F_{m_*}(\lambda)\left\langle K_m r_\lambda,r_\mu\right\rangle,
$$
from which, using again Lemma \ref{nondegeneracylemma}, we conclude that $\left\langle K_m r_\lambda,r_\mu\right\rangle$ vanishes for all $\lambda\not=\mu$. This implies, as $r_\lambda(T;\s)$ form a basis, that
$$
K_m r_\lambda(T;\s)=F_m(\lambda;\s)r_\lambda(T;\s)
$$
for some scalar-valued series $F_m(\lambda;\s)=\sum_{k\geq 0}F_m^{[k]}(\lambda)\s^k$. 
The proof is complete. $\hfill\square$

In the above proof, the nondegeneracy property given in Lemma~\ref{nondegeneracylemma} 
is used. We note that a weaker form of the nondegeneracy property 
(namely that for a fixed $m\ge 1$, if $E_m^{[0]}(\lambda)=E_m^{[0]}(\mu)$ for some $\lambda\neq \mu$ 
then there exists $k>m$ such that $E_k^{[0]}(\lambda)\neq E_k^{[0]}(\mu)$)
is actually sufficient for the proof, and that we can prove this weaker property for the quantum KdV Hamiltonians $H_m$ even 
when $H_m$ are restricted to $U_0=0$. This could also be helpful for other quantum integrable systems.  

A natural question for the quantum KdV hierarchy is 
about the analyticity in~$\sigma$ of their spectrum and common eigenvectors. 
A way to answer this question is to look at the spectral curves~\eqref{spectralcurves} (see the Introduction). 
Indeed, the spectral curves are ramified coverings of the Riemann sphere~$\s\in\mathbb P^1$ of degree~$|\Y_k|$. 
So the eigenvalues $F_m(\lambda;\sigma)$ above are the Taylor series at $\sigma=0$ of the 
branches of the analytic functions~$\rho$ on~$\Sigma_{k,m}$.
Similarly, $r_{\lambda\mu}(\sigma)$  
are meromorphic functions on the 
spectral curves~$\Sigma_{k,m}$. 
The global geometry of the spectral curves is interesting. 
For example, it is possible to prove that at $\s=\infty$ the eigenvectors are the monomials $T_\ll=T_{\ll_1}\cdots T_{\ll_{\ell(\ll)}}$ 
and the corresponding branches of~$\rho$ behave as 
\be
\rho \sim \s^{m}\sum_{i=1}^{\ell(\ll)}\ll_i^{2m+1}, \quad \sigma\to\infty,
\ee
while, from Theorem~\ref{thmdubrovin} we know that at $\s=0$ 
the eigenvectors are the Schur polynomials $s_\ll(T)$;
thus, such a deformation of Schur polynomials interpolates between the Schur and monomial bases of $\wt\Lambda$.
Moreover, due to the symmetry property of the Hamiltonians, the branch points can only be complex conjugate points in the $\s$-plane; coincidence of eigenvalues for real values of $\s$ yields instead singularities of the spectral curves. Further study of the geometry of spectral curves is deferred to future investigations.

\begin{remark}
The initial value problem
\begin{align}
\hbar\frac{\pa}{\pa y_m}\Psi(\underline{y})&=H_m\Psi(\underline{y}),\quad m\geq 0,
\\
\Psi|_{\underline{y}=0}&=\sum_{\ll\in\Y} c_\ll r_\ll\biggl(\frac q{\sqrt\hbar};-\frac{\e^2}{\sqrt\hbar}\biggr),
\end{align}
where $c_\ll$ are arbitrarily given constants and $\underline{y}=(y_0,y_1,y_2,\dots)$, can be solved as
\be
\Psi(\underline{y})=\sum_{\ll\in\Y} c_\ll{\rm e}^{\frac 1\hbar\sum_{m\geq 0} y_m E_m(\ll;\e,\hbar,U_0)}r_\ll\biggl(\frac q{\sqrt\hbar};-\frac{\e^2}{\sqrt\hbar}\biggr).
\ee
\end{remark}

The classical property
\be
\label{schurconjugate}
s_{\ll'}(T)=(-1)^{|\ll|}s_{\ll}(-T)
\ee
of the Schur polynomials, where $\ll'$ denotes the conjugate partition of~$\ll$ 
(i.e. the diagram of~$\ll'$ is obtained by flipping that of~$\lambda$ along its main diagonal), 
generalizes to a similar property of the deformed Schur polynomials.
\begin{shaded}
\begin{proposition}
\label{propositionconjugate}
For all $\ll\in\Y$ we have
\be
r_{\ll'}(T;\s)=(-1)^{|\ll|}r_{\ll}(-T;-\s).
\ee
\end{proposition}
\end{shaded}
\begin{proof}
Let us denote $\wt r_\ll(T;\s)=(-1)^{|\ll|}r_{\ll'}(-T;-\s)$; from \eqref{schurconjugate} we obtain
\be
\label{eq:1proof}
\left\langle \wt r_{\ll}(T;\s),s_{\ll}(T) \right\rangle=\left\langle r_{\ll'}(-T;-\s),s_{\ll'}(-T) \right\rangle=1,
\ee
where we also use the property~\eqref{1} of deformed Schur polynomials.
Next, it follows from Remark~\ref{remarkhomogeneous} that all quantum KdV Hamiltonians are homogeneous of degree $m+2$ with respect to the degree assignment
\be
\label{homogeneityproof}
\wt\deg\, q_k=1,\qquad \wt\deg\,\e=\frac 12,\qquad\wt\deg\,\hbar=2.
\ee
Let us denote by $H_m=H_m(\e,\hbar)$ the explicit dependence 
on the parameters and let us also introduce the involution $\Pi:\L\mapsto\L$ defined by $q_i\mapsto -q_i$ 
on the generators of the polynomial ring~$\L$.
Thus, from the homogeneity property just mentioned (see \eqref{homogeneityproof}) we have
\be
\label{identityproof}
\Pi H_m(\i\e,\hbar)\Pi=(-1)^{m}H_m(\e,\hbar).
\ee
It follows then that $H_m(\i\e,\hbar)$ is diagonal 
on $\Pi r_\ll(q/\sqrt\hbar;\e^2/\sqrt\hbar)=r_\ll(-q/\sqrt\hbar;\e^2/\sqrt\hbar)$, 
hence it is diagonal on $\wt r_\ll(q/\sqrt\hbar;-\e^2/\sqrt\hbar)$.
Thus the elements $\wt r_{\ll}(T;\s)$ satisfy the two defining properties \eqref{1} and \eqref{2} uniquely characterizing the deformed Schur polynomials (cf. Proposition \ref{prop:existencedeformedschur}) and so they must be equal to $r_\ll(T;\s)$.
\end{proof}

We give a few more examples of $r_\ll(T;\s)$.
{\footnotesize \begin{align*}
r_{(3)}&=s_{(3)}+s_{(2,1)}\left(-\frac{2\s}{9}+\frac{\s^2}{324}+\frac{43 \s^3}{5832}
+\frac{193 \s^4}{559872}+\O(\s^{5})\right)+s_{(1,1,1)}\left(
\frac{5 \s}{72}+\frac{2 \s^2}{81}-\frac{893 \s^3}{373248}-\frac{115 \s^4}{69984}+\O(\s^{5})\right),
\\
r_{(2,1)}&=s_{(3)}\left(\frac{2 \s}{9}+\frac{\s^2}{81}-\frac{2 \s^3}{729}-\frac{\s^4}{729}+\O(\s^{5})\right)
+s_{(2,1)}
+s_{(1,1,1)}\left(-\frac{2 \s}{9}+\frac{\s^2}{81}+\frac{2 \s^3}{729}-\frac{\s^4}{729}+\O(\s^{5})\right),
\end{align*}}

\vspace{-6mm}

{\footnotesize \begin{align*}
r_{(4)}&=s_{(4)}+s_{(3,1)}\left(-\frac{5 \s}{16}+\frac{\s^2}{192}+\frac{6055 \s^3}{331776}+\O(\s^4)\right)
+s_{(2,2)}\left(-\frac{5 \s}{72}+\frac{59 \s^2}{2592}+\frac{4715 \s^3}{1492992}+\O(\s^4)\right)
\\
&\qquad
+s_{(2,1,1)}\left(\frac{\s}{8}+\frac{37 \s^2}{768}-\frac{727 \s^3}{82944}+\O(\s^4)\right)
+s_{(1,1,1,1)}\left(-\frac{7 \s}{144}-\frac{95 \s^2}{2592}-\frac{9119 \s^3}{2985984}+\O(\s^4)\right),
\\
r_{(3,1)}&=s_{(4)}\left(\frac{5 \s}{16}+\frac{\s^2}{32}-\frac{7 \s^3}{4096}+\O(\s^4)\right)
+s_{(3,1)}
+s_{(2,2)}\left(-\frac{\s}{8}-\frac{\s^2}{32}-\frac{13 \s^3}{2048}+\O(\s^4)\right)
\\
&\qquad
+s_{(2,1,1)}\left(-\frac{5 \s}{16}+\frac{3 \s^2}{64}+\frac{35 \s^3}{4096}+\O(\s^4)\right)
+s_{(1,1,1,1)}\left(\frac{\s}{8}+\frac{11 \s^2}{256}-\frac{7 \s^3}{1024}+\O(\s^4)\right),
\\
r_{(2,2)}&=s_{(4)}\left(\frac{5 \s}{72}+\frac{37 \s^2}{1296}-\frac{133 \s^3}{46656}+\O(\s^4)\right)
+s_{(3,1)}\left(\frac{\s}{8}-\frac{\s^2}{48}-\frac{31 \s^3}{5184}+\O(\s^4)\right)
+s_{(2,2)}
\\
&\qquad
+s_{(2,1,1)}\left(-\frac{\s}{8}-\frac{\s^2}{48}+\frac{31 \s^3}{5184}+\O(\s^4)\right)
+s_{(1,1,1,1)}\left(-\frac{5 \s}{72}+\frac{37 \s^2}{1296}+\frac{133 \s^3}{46656}+\O(\s^4)\right).
\end{align*}}

\subsection{Vanishing identities}\label{sectionvanishingid}

As a consequence of Proposition~\ref{prop:existencedeformedschur}, we get some vanishing identities for certain combination of characters in the symmetric group. For convenience, let us denote by~$\Y_k$ the set of partitions $\ll$ of weight $|\ll|=k$ and set
\be
\wt H_m^{[i]}:=\left.H_m^{[i]}\right|_{U_0=0},\qquad\wt E^{[0]}_m(\ll;\hbar)=\left.E^{[0]}_m(\ll;\hbar,U_0)\right|_{U_0=0}.
\ee
By \eqref{dispersionlesseigenvalues} we infer that $\wt E^{[0]}$ is given by
\be
\label{EQ}
\wt E^{[0]}_m(\ll;\hbar)=\sqrt{\hbar^{m+2}}Q_{m+2}(\ll),
\ee
where $Q_j$ are given in \eqref{Q}.

\begin{lemma}\label{vanishing1}
Suppose $\ll,\mu\in\Y_k$ are distinct partitions of the same weight, $\ll\not=\mu$, satisfying $\wt E_m^{[0]}(\ll;\hbar)=\wt E_m^{[0]}(\mu;\hbar)$.
Then
\be
\left\langle s_\ll\left(\frac q{\sqrt\hbar}\right),\wt H_m^{[1]}s_\mu\left(\frac q{\sqrt\hbar}\right)\right\rangle=0.
\ee
\end{lemma}

\begin{proof}
Taking terms of order $\epsilon$ in the identity $[H_m,H_{m'}]=0$ we have, after setting $U_0=0$,
\be
0=[\wt H_m^{[0]},\wt H_{m'}^{[1]}]+[\wt H_m^{[1]},\wt H_{m'}^{[0]}].
\ee
Taking the matrix entry $(\ll,\mu)$ of this relation, using the symmetry property of Corollary \ref{corollarysymmetry}, and denoting $T=q/\sqrt\hbar$ we get
\begin{align}
0&=\left\langle s_\ll(T),\left([\wt H_m^{[0]},\wt H_{m'}^{[1]}]+[\wt H_m^{[1]},\wt H_{m'}^{[0]}]\right)s_\mu(T)\right\rangle \nonumber
\\&=\left\langle \wt H_m^{[0]}s_\ll(T),\wt H_{m'}^{[1]}s_\mu(T)\right\rangle
-\left\langle \wt H_{m'}^{[1]}s_\ll(T),\wt H_m^{[0]}s_\mu(T)\right\rangle
\nonumber\\
&\quad+\left\langle \wt H_m^{[1]}s_\ll(T),\wt H_{m'}^{[0]}s_\mu(T)\right\rangle
-\left\langle \wt H_{m'}^{[0]}s_\ll(T),\wt H_m^{[1]}s_\mu(T)\right\rangle \nonumber
\\&=\left(\wt E_m^{[0]}(\ll;\hbar)-\wt E_m^{[0]}(\mu;\hbar)\right)\left\langle s_\ll(T),\wt H_{m'}^{[1]}s_\mu(T)\right\rangle-
\left(\wt E_{m'}^{[0]}(\ll;\hbar)-\wt E_{m'}^{[0]}(\mu;\hbar)\right)\left\langle s_\ll(T),\wt H_m^{[1]}s_\mu(T)\right\rangle.
\end{align}
Now assume $\ll,\mu$ are as in the statement; using the first part in the proof of Lemma \ref{nondegeneracylemma}, we must have $\wt E_{m'}^{[0]}(\ll;\hbar)\not=\wt E_{m'}^{[0]}(\mu;\hbar)$ for some $m'$, hence we have $\left\langle s_\ll(\wt q),\wt H_m^{[1]}s_\mu(\wt q)\right\rangle=0$.
\end{proof}

Rephrasing the statement of Lemma~\ref{vanishing1} for $m=1$, we get the more explicit identity given in the following corollary.

\begin{shaded}
\begin{corollary}
\label{corollaryvanishing1}
Suppose $\ll,\mu\in\Y_k$ are distinct partitions, $\ll\not=\mu$, satisfying
$P_2(\ll)=P_2(\mu)$, where $P_2$ is defined in \eqref{P}. Then
\be
\label{final}
\sum_{\nu\in\Y_k}|C_\nu|\chi_\lambda(C_\nu)\chi_\mu(C_\nu)\sum_{i=1}^{\ell(\nu)} \nu_i^3=0,
\ee
where $C_\nu\subset\mathfrak S_k$ is the conjugacy class of permutations with disjoint cycles of lengths $\nu_1,\dots,\nu_{\ell(\nu)}$
and $\chi_\ll$ is the character of the irreducible representation of $\mathfrak S_k$ associated to $\ll$.
\end{corollary}
\end{shaded}

\begin{proof}
By using the explicit formula (see e.g. \cite[Section 4.1]{BR2016a})
\be
\left.h_1\right|_{U_0=0}=\frac{u^3}6+\frac{\e^2}{24}uu_{xx}-\frac{\hbar\e^2}{2880},
\ee
we obtain
\be
\left.H_1\right|_{U_0=0}=\Delta-\frac{\hbar\e^2}{12}\Biggl(\sum_{i\geq 1}i^3q_i\frac{\pa}{\pa q_i}+\frac{1}{240}\Biggr),
\ee
where $\Delta$ is the cut-and-join operator, see~\eqref{cutandjoin}. Hence, since $\ll\not=\mu$ are assumed to satisfy $P_2(\ll)=P_2(\mu)$, we have $\wt E_1^{[0]}(\ll;\hbar)=\wt E_1^{[0]}(\mu;\hbar)$, cf.~\eqref{Q} and~\eqref{EQ}, and so Lemma \ref{vanishing1} yields
\be
0=\left\langle s_\ll\biggl(\frac q{\sqrt\hbar}\biggr),\wt H_1^{[1]}s_\mu\biggl(\frac q{\sqrt\hbar}\biggr)\right\rangle=\left\langle 
s_\ll\biggl(\frac q{\sqrt\hbar}\biggr),\sum_{i\geq 1}i^3q_i\frac{\pa}{\pa q_i}s_\mu\biggl(\frac q{\sqrt\hbar}\biggr)\right\rangle.
\ee
By expanding the Schur polynomials on the basis of monomials with the help of 
\be
\label{schurvscharacters}
s_\ll\biggl(\frac q{\sqrt\hbar}\biggr)=\sum_{\nu\in\Y_k}\frac{|C_\nu| \, \chi_\ll(C_\nu)}{k!}\frac{q_\nu}{\hbar^{\ell(\nu)/2}},
\ee
(see \cite{M2015}) and by using \eqref{scalarscalarproduct}, the last equation can be rewritten as \eqref{final}.
\end{proof}


\section{Further remarks}\label{furtherremarks}
This paper is a first step in the study of the spectral problems
for the quantum KdV hierarchy following Dubrovin's suggestion~\cite{D2016,Dubrovin}.
In particular, we have proved that the quantum KdV Hamiltonians are symmetric with respect to a natural inner product
(this property is actually generalized and proved for any 
tautological CohFT; cf.~Corollary~\ref{corollarysymmetry}) 
and that the quantum dispersionless KdV hierarchy possesses a non-degeneracy property.
Moreover, we show that these two properties imply the existence of a complete set of common
eigenvectors for the quantum KdV Hamiltonians.  
As an application, we obtain some vanishing identities for certain combinations of characters in the symmetric group; 
a simple example of this phenomenon is given in Corollary \ref{corollaryvanishing1}.
More applications will be given in subsequent publications.

Let us note that 
Proposition \ref{prop:existencedeformedschur} can be generalized to the 
 quantum DR hierarchy associated to an arbitrary rank-1 CohFT (cf.~\cite{BDGR2019,BR2016a, DLYZ2016,T2012}). 
More precisely, recall that the quantum Hamiltonian densities for this hierarchy are
\be
h_{m}^{\sf Hodge}(\underline{s})=\sum_{\begin{smallmatrix}g,n\geq 0\\ 2g-2+n\geq 0\\ k_1,\dots,k_n\in\mathbb{Z}
\end{smallmatrix}}
\frac{\hbar^g}{n!}\left(\int_{DR_g\left(-\sum\limits_{i=1}^n k_i,k_1,\dots,k_n\right)}\psi_1^m\L\left(-\frac{\epsilon^2}{\hbar}\right)
{\rm e}^{\sum_{j\geq 1}{\rm ch}_{2j-1}s_{2j-1}}
\right) U_{k_1}\cdots U_{k_n} {\rm e}^{\i x\sum\limits_{i=1}^nk_i},
\ee
where $\underline{s}=(s_1,s_3,s_5,\dots)$ and ${\rm ch}_{2j-1}$ denotes the $(2j-1)$th 
component of the Chern character of the Hodge bundle on the moduli space of curves.
We have the following proposition, where we denote $H^{\sf Hodge}_m=\wh{\overline{h_{m}^{\sf Hodge}}}$ the associated quantum DR Hamiltonians

\begin{shaded}
\begin{proposition}
\label{prop:existencedeformedschurHodge}
There exists a unique collection $\left(r_\ll(T;\sigma,\underline{s})\right)_{\ll\in\Y}$ of elements in the free module $\C[T_1,T_2,\dots]\otimes_\C\C\llbracket\sigma,\underline{s}\rrbracket$ over $\C\llbracket\sigma,\underline{s}\rrbracket$ such that
\begin{align}
\label{1Hodge}
\left\langle r_\ll(T;\sigma,\underline{s}),s_\ll(T)\right\rangle&=1,
\\
\label{2Hodge}
H_m^{\sf Hodge} r_\ll\left(\frac q{\sqrt\hbar};-\frac{\e^2}{\sqrt\hbar},\underline{s}\right)&=E_m\left(\ll;\e,\hbar,\underline{s}\right)r_\ll\left(\frac q{\sqrt\hbar};-\frac{\e^2}{\sqrt\hbar},\underline{s}\right),
\end{align}
for some $E_m(\ll;\e,\hbar,\underline{s})\in\C(\sqrt\hbar)\llbracket\e^2,\underline{s}\rrbracket$. Moreover, $\left(r_\ll(T;\sigma,\underline{s})\right)_{\ll\in\Y}$ is basis of the free module $\C[T_1,T_2,\dots]\otimes_\C\C\llbracket\sigma,\underline{s}\rrbracket$, and 
\be
\label{3Hodge}
r_\ll(T;\sigma=0,\underline{s}=0)=s_\ll(T),\qquad E_m(\ll;\e=0,\hbar,\underline{s}=0)=E_m^{[0]}(\ll;\hbar),
\ee 
where $E_m^{[0]}$ are defined in \eqref{dispersionlesseigenvalues}.
\end{proposition}
\end{shaded}

The proof of this proposition is analogous to that of Proposition~\ref{prop:existencedeformedschur} and is therefore omitted.

We end this paper by giving an equivalent form of Buryak--Rossi's quantum DR recursion in the rank-1 case.
Indeed, by using the dimensional constraint in \eqref{2Hodge}, we infer that $h_{m}^{\sf Hodge}$ is homogeneous of degree $m+2$ with respect to the grading
\be
\label{quantumKodamaHodge}
\wt\deg\,u=1,\qquad\wt{\deg}\,\e=\frac 12,\qquad\wt{\deg}\,\hbar=2,\qquad\wt{\deg}\,s_{2k-1}=-2k+1.
\ee
In particular, this implies the relation
\be
\biggl(\frac \e 2\pa_\e+2\hbar\pa_\hbar-\sum_{k\geq 1}(2k-1)s_{2k-1}\pa_{s_{2k-1}}+\sum_{k\geq 0}u_{kx}\pa_{u_{kx}}\biggr)h_m^{\sf Hodge}=(m+2)h_m^{\sf Hodge},
\ee
which can be combined with Buryak--Rossi's recursion \eqref{BRrec1} to give an alternative recursion
\be
\biggl(m+1+\frac \e 2\frac\pa{\pa\e}+\sum_{k\geq 1}(2k-1)s_{2k-1}\pa_{s_{2k-1}}\biggr)\pa_xh_m^{\sf Hodge}=\frac 1\hbar\left[ h_{m-1}^{\sf Hodge},H_1^{\sf Hodge}\right].
\ee
In terms of the generating function $\mathcal{H}^{\sf Hodge}(z):=\sum_{m\geq -1}z^{m+1}h_m^{\sf Hodge}$ it gives
\be\label{MCH}
\mathscr D\,\mathcal{H}^{\sf Hodge}=0,\qquad \mathscr D=z\frac\pa{\pa z}+\frac \e 2\frac\pa{\pa\e}+\sum_{k\geq 1}(2k-1)s_{2k-1}\pa_{s_{2k-1}}+\frac{z\pa_x^{-1}}{\hbar}\left[H_1^{\sf Hodge},\cdot\right].
\ee
We note that a similar equation to~\eqref{MCH} is also given in~\cite{BDGR2019} (cf.~Theorem~5.1 therein). Taking the $\hbar\to0$ limit in~\eqref{MCH} 
we find
\be
\mathscr D^{\sf cl} \,\mathcal{H}^{\sf Hodge,cl}=0,\qquad \mathscr D^{\sf cl}=z\frac\pa{\pa z}+\frac \e 2\frac\pa{\pa\e}+\sum_{k\geq 1}(2k-1)s_{2k-1}\pa_{s_{2k-1}}+z\pa_x^{-1}\left\{H_1^{\sf Hodge,cl},\cdot\right\}.
\ee

\appendix

\section{Quantum Hopf hierarchy and Young-Jucys-Murphy elements}\label{AppYJM}

As a consequence of Theorem~\ref{thmdubrovin} \cite{D2016} we prove in this appendix a representation of the quantum dispersionless KdV Hamiltonians $H_m^{[0]}$ in terms of multiplication operators in the class algebra of the symmetric groups.

With respect to the gradation $\deg q_k=k$, consider the direct sum decomposition $\L=\bigoplus_{k\geq 0}\Lambda_k$ into homogeneous components.
Denote by $\mathfrak S_k$ the symmetric group of permutations of $\{1,\dots,k\}$, by $\C[\mathfrak S_k]$ its group algebra, and introduce the \emph{Frobenius map}\footnote{It is convenient here to include the quantization parameter $\hbar$, though it may be absorbed by the scaling $T=q/\sqrt\hbar$.}
\begin{align}
\Phi:\C\left[\mathfrak S_k\right]&\to\L_k\left[1/\sqrt\hbar\right]
\\
\label{frobenius}
\s&\mapsto\frac{q_\ll}{\hbar^{\ell/2}k!}\, ,
\end{align}
defined here for permutations $\s$ with $\ell$ disjoint cycles of lengths $\ll_1.\dots,\ll_{\ell}$ and extended linearly, denoting by $q_\ll$ the monomial $q_\ll:=q_{\ll_1}\cdots q_{\ll_\ell}$.

It is well known \cite{G1994} that $\Phi^{-1}\Delta\Phi=\Phi^{-1}H_1^{[0]}\big|_{U_0=0}\Phi$, see \eqref{cutandjoin}, is the operator in the class algebra $Z(\C[\mathfrak S_k])$ of the symmetric group (the center of the group algebra $\C[\mathfrak S_k]$) given by multiplication by the formal sum of all transpositions.
Based on Dubrovin's result (Theorem \ref{thmdubrovin}) we can prove the following result, concerning an extension of this representation-theoretic interpretation to the entire sequence of the quantum dispersionless KdV Hamiltonians.
To this end, let us introduce the \emph{Young-Jucys-Murphy} (YJM) elements $\mathcal J_i\in\C[\mathfrak S_i]$ \cite{J1974,M1981,OV1996}, defined as the sum of all transpositions in $\mathfrak S_i$ which are not in $\mathfrak S_{i-1}$. Explicitly:
\be
\label{YJM}
\mathcal J_1=0,\qquad \mathcal J_i=(1,i)+(2,i)+\cdots+(i-1,i),\quad i>1.
\ee

\begin{shaded}
\begin{proposition}
\label{propYJM}
For any $k\geq 0$, we have the identity
\be\label{propYJMeq}
1+\sum_{m\geq -1}\left(\frac z{\sqrt\hbar}\right)^{m+2}\left.H^{[0]}_m\right|_{\L_k}={\rm e}^{z U_0/\sqrt\hbar}\left(\frac {z/2}{\sinh(z/2)}+2 z\sinh\left(\frac z2\right)\sum_{m\geq 0}\frac{z^m}{m!}\Phi\left(\mathcal J_1^m+\dots+\mathcal J_k^m\right)\Phi^{-1}\right)
\ee
in ${\rm End}_\C(\L_k)\otimes\C((\sqrt\hbar))\llbracket z\rrbracket$. Here $\mathcal J_i$ are the YJM elements \eqref{YJM} and $\Phi$ the Frobenius map \eqref{frobenius}.
\end{proposition}
\end{shaded}

Let us remark that Rossi \cite[Theorem 5.3]{R2008} gives an alternative representation of these operators in terms of free fermions.

Before the proof, we note that for rank-1 CohFTs, the sesquilinear form \eqref{defscalarproduct} reduces to
\be
\left\langle q_\ll,q_\mu\right\rangle=\hbar^{\ell(\ll)}z_\ll\delta_{\ll,\mu},
\ee
where $z_\ll:=\prod_{i\geq 1}m_i(\ll)!i^{m_i(\ll)}$, denoting $m_i(\ll)$ the number of parts of $\ll$ which are equal to $i$, and $\ell(\ll)$ the length of $\ll$.
Equivalently, denoting $C_\ll\subset\mathfrak S_{|\ll|}$ the conjugacy class of permutations with disjoint cycles of lengths $\ll_1,\dots,\ll_{\ell(\ll)}$, we have
\be
\label{scalarscalarproduct}
\left\langle \frac{q_\ll}{\hbar^{\ell(\ll)/2}},\frac{q_\mu}{\hbar^{\ell(\mu)/2}}\right\rangle=\frac{|\lambda|!}{|C_\ll|}\delta_{\ll,\mu}.
\ee
This is the standard inner product on the space of symmetric polynomials \cite{M2015}; in particular the Schur basis is orthonormal
\be
\label{orthonormalityschur}
\left\langle s_\ll\biggl(\frac q{\sqrt\hbar}\biggr),s_\mu\biggl(\frac q{\sqrt\hbar}\biggr)\right\rangle=\delta_{\ll,\mu}.
\ee

\smallskip

\noindent{\it Proof of Proposition~\ref{propYJM}.}
Before entering the proof proper let us recall the following two fundamental properties of the YJM elements introduced in \eqref{YJM}.
\begin{itemize}
\item[{\em (i)}] The $\mathcal J_i$'s commute among themselves, though they are not in the class algebra $Z(\C[\mathfrak S_k])$ (the center of the group algebra); however the class algebra is generated by symmetric functions of the YJM elements.
\item[{\em (ii)}] Multiplication by a symmetric function $f(\mathcal J_1,\dots,\mathcal J_k)$ of the YJM elements is diagonal on the character basis $\{\chi_\ll\}_{\ll\in\Y_k}$ of the class algebra $Z(\C[\mathfrak S_k])$, with explicit eigenvalues 
\be
\label{eigenvalueYJM}
f(\mathcal J_1,\dots,\mathcal J_k)\chi_\ll=f(\{i-j\}_{(i,j)\in\ll})\chi_\ll, \qquad \ll\in\Y_k,
\ee
where $f(\{i-j\}_{(i,j)\in\ll})$ is the evaluation of the symmetric function $f$ at the \emph{contents} $i-j$ of the partition $\ll$.
Here the partition $\ll$ is identified with the set of $(i,j)\in\mathbb Z^2$ satisfying $1\leq i\leq\ell(\ll)$, $1\leq j\leq\ll_i$.
\end{itemize}

Therefore to prove \eqref{propYJMeq}, let us note that, since $\Phi\chi_\ll=s_\ll(q/\sqrt\hbar)$ \cite{M2015}, both sides act diagonally on the Schur basis, and so it suffices to prove that they have the same eigenvalues.

To this end it is convenient to employ the Frobenius notation for a partition $\ll=(\ll_1,\dots,\ll_\ell)$, which consists in denoting
\be
\ll=(a_1,\dots,a_d|b_1,\dots,b_d),
\ee
where $d:=\max\{i\geq 0:\ \ll_i\geq i\}$ is the length of the main diagonal in the Young diagram of $\ll$, and for $i=1,\dots,d$,
\be
a_i:=\ll_i-i,\qquad b_i=\ll'_i-i
\ee
are the number of cells to the right of the $i$th diagonal cell in the Young diagram of $\ll$ ($a_1>\dots>a_d\geq 0$), and the number of cells below the $i$th diagonal cell in the Young diagram of $\ll$ ($b_1>\dots>b_d\geq 0$), respectively; here, as before, we denoted $\ll'$ the conjugate partition.

Hereafter we denote $\ll=(a_1,\dots,a_d|b_1,\dots,b_d)$ in Frobenius notation, and using \eqref{eigenvalueYJM} we compute
\be
(\mathcal J_1^m+\dots+\mathcal J_k^m)\chi_\ll=\Biggl(\sum_{(i,j)\in\ll}(i-j)^m\Biggr)\chi_\ll
=\sum_{i=1}^{d}\Biggl(\delta_{m,0}+\sum_{\ell=1}^{a_i}\ell^m+(-1)^m\sum_{\ell=1}^{b_i}\ell^m\Biggr)\chi_\ll
\ee
and using the \emph{Faulhaber-Bernoulli formula} (for $m=0,1,2,\dots$)
\be
\sum_{\ell=1}^N\ell^m=\frac{F_{m+1}(N)}{m+1},\qquad F_m(x)=\sum_{j=0}^{m-1}(-1)^j\binom{m}{j}B_jx^{m-j},
\ee
we can write
\be
(\mathcal J_1^m+\dots+\mathcal J_k^m)\chi_\ll=\frac 1{m+1}\sum_{i=1}^{d}\left(\delta_{m,0}+F_{m+1}(a_i)+(-1)^mF_{m+1}(b_i)\right)\chi_\ll.
\ee
Multiplying by $z^m/m!$ and summing over $m\geq 0$ we obtain
\begin{align}
\nonumber
\sum_{m\geq 0}\frac{z^m}{m!}(\mathcal J_1^m+\dots+\mathcal J_k^m)\chi_\ll&=\sum_{m\geq 0}\frac{z^m}{(m+1)!}\sum_{i=1}^{d}\left(\delta_{m,0}+F_{m+1}(a_i)+(-1)^mF_{m+1}(b_i)\right)\chi_\ll
\\
\label{eigenvalue1}
&=\sum_{i=1}^{d}\left(1+\frac{{\rm e}^{za_i}-1}{1-{\rm e}^{-z}}+\frac{{\rm e}^{-zb_i}-1}{1-{\rm e}^z}\right)\chi_\ll,
\end{align}
where we have used the identity
\be
\sum_{m\geq 0}\frac{F_{m+1}(x)}{(m+1)!}z^m=\frac{{\rm e}^{zx}-1}{1-{\rm e}^{-z}}.
\ee
Finally, from the results of \cite{D2016}, recalled in Theorem \ref{thmdubrovin}, we have
\be
\label{eigenvalue2}
\mathcal{H}^{[0]}\left(\frac z{\sqrt\hbar}\right)s_\ll\left(\frac q{\sqrt\hbar}\right)={\rm e}^{z U_0/\sqrt\hbar}\left(\frac 1{S(z)}+z\sum_{i=1}^{d}\left({\rm e}^{z\left(a_i+\frac 12\right)}-{\rm e}^{-z\left(b_i+\frac 12\right)}\right)\right)s_\ll\left(\frac q{\sqrt\hbar}\right),
\ee
where the generating function $\mathcal H^{[0]}$ is given in \eqref{Eliashbergformula} and $S(z)$ is given in \eqref{S}.
Therefore both sides of \eqref{propYJMeq} are diagonal on Schur polynomials $s_\ll\left(q/\sqrt\hbar\right)$, and by the identity
\be
{\rm e}^{z\left(a+\frac 12\right)}-{\rm e}^{-z\left(b+\frac 12\right)}=2\sinh\left(\frac z2\right)\left(\frac{{\rm e}^{az}-1}{1-{\rm e}^{-z}}+\frac{{\rm e}^{-bz}-1}{1-{\rm e}^{z}}+1\right)
\ee
it is clear that the eigenvalues coincide, see \eqref{eigenvalue1} and \eqref{eigenvalue2}, and the proof is complete.
\hfill $\square$

\begin{remark}
In view of the identification of Prop.~\ref{propYJM}, Eliashberg's formula~\eqref{Eliashbergformula} matches with a formula of Lascoux and Thibon \cite{LT2001} expressing the action of multiplication in the class algebra by Newton polynomials of YJM elements in terms of differential operators.
\end{remark}

It remains an open question whether a similar description continues to hold for the quantum KdV hierarchy; 
in this respect it would be interesting to compare with the results of the recent
preprint~\cite{Blot} where Hurwitz numbers appear in connection with the quantum Witten-Kontsevich
series, a particular {\it quantum tau-function} for the quantum KdV hierarchy.


\begin{thebibliography}{12}
\bibitem{Bazhanov1}
V.~Bazhanov, S.~Lukyanov, and A.~Zamolodchikov.
\newblock ``Integrable structure of conformal field theory, quantum KdV theory and thermodynamic Bethe ansatz''.
\newblock {\em Comm. Math. Phys.} 177 (1996) 381--398.

\bibitem{Bazhanov2}
V.~Bazhanov, S.~Lukyanov, and A.~Zamolodchikov.
\newblock ``Integrable structure of conformal field theory II. Q-operator and DDV equation''.
\newblock {\em Comm. Math. Phys.} 190 (1997) 247--278.

\bibitem{Bazhanov3}
V.~Bazhanov, S.~Lukyanov, and A.~Zamolodchikov.
\newblock ``Integrable structure of conformal field theory III. The Yang-Baxter relation''.
\newblock {\em Comm. Math. Phys.} 200 (1999) 297--324.

\bibitem{BO2000}
S.~Bloch and A.~Okounkov.
\newblock ``The character of the infinite wedge representation''.
\newblock {\em Adv. Math.} 149 (2000), no. 1, 1--60. 

\bibitem{Blot}
X.~Blot.
\newblock ``The quantum Witten-Kontsevich series and one-part double Hurwitz numbers''.
\newblock Preprint {\em arXiv:2004.07581}.

\bibitem{BSTV2016}
G.~Bonelli, A.~Sciarappa, A.~Tanzini, and P.~Vasko.
\newblock ``Six-dimensional supersymmetric gauge theories, quantum cohomology of instanton moduli spaces and gl(N) Quantum Intermediate Long Wave Hydrodynamics''.
\newblock {\em J. High Energ. Phys.} 2014, 141 (2014).

\bibitem{B2015}
A.~Buryak.
\newblock ``Double ramification cycles and integrable hierarchies''.
\newblock {\em Comm. Math. Phys.} 336 (2015), no. 3, 1085--1107.

\bibitem{BDGR2019}
A.~Buryak, B.~Dubrovin, J.~Gu\'er\'e, and P.~Rossi.
\newblock ``Integrable Systems of Double Ramification Type''.
\newblock {\em IMRN} 2020 (2020), no. 24, 10381--10446.

\bibitem{BPS1}
A.~Buryak, H.~Posthuma, and S.~Shadrin. 
\newblock ``A polynomial bracket for the Dubrovin-Zhang hierarchies". 
\newblock {\em J. Differential Geom.} 92 (2012), no. 1, 153--185.

\bibitem{BPS2}
A.~Buryak, H.~Posthuma, and S.~Shadrin. 
\newblock ``On deformations of quasi-Miura transformations and the Dubrovin-Zhang bracket". 
\newblock {\em J. Geom. Phys.} 62 (2012), no. 7, 1639--1651.

\bibitem{BR2016a}
A.~Buryak and P.~Rossi.
\newblock ``Double ramification cycles and quantum integrable systems''.
\newblock {\em Lett. Math. Phys.} 106 (2016), no. 3, 289--317.

\bibitem{BR2016b}
A.~Buryak and P.~Rossi.
\newblock ``Recursion relations for double ramification hierarchies''.
\newblock {\em Comm. Math. Phys.} 342 (2016), no. 2, 533--568.

\bibitem{BSSZ2015}
A.~Buryak, S.~Shadrin, L.~Spitz, and D.~Zvonkine.
\newblock ``Integrals of $\psi$-classes over double ramification cycles''.
\newblock {\em Amer. J. Math.} 137 (2015), no. 3, 699--737.

\bibitem{CMZ2018}
D.~Chen, M.~M\"oller, and D.~Zagier.
\newblock ``Quasimodularity and large genus limits of Siegel-Veech constants''.
\newblock {\em J. Amer. Math. Soc.} 31 (2018), no. 4, 1059--1163. 

\bibitem{DMS2005} 
L.~Degiovanni, F.~Magri, and V.~Sciacca. 
\newblock ``On deformation of Poisson manifolds of hydrodynamic type''.
\newblock {\em Comm. Math. Phys.} 253 (2005), 1--24.

\bibitem{Dubrovin2DTFTs}
B.~Dubrovin.
\newblock ``Geometry of 2D topological field theories''. \emph{Integrable systems and quantum groups (Montecatini Terme, 1993)}, 120--348.
\newblock Lecture Notes in Math., 1620, Fond. CIME/CIME Found. Subser., {\em Springer}, {\em Berlin}, 1996. 

\bibitem{D2016}
B.~Dubrovin.
\newblock ``Symplectic field theory of a disk, quantum integrable systems, and Schur polynomials''.
\newblock {\em Ann. Henri Poincar\'e} 17 (2016), no. 7, 1595--1613.

\bibitem{Dubrovin}
B.~Dubrovin. 
\newblock Private communication (2016). 

\bibitem{DLYZ2016}
B.~Dubrovin, S.-Q.~Liu, D.~Yang, and Y.~Zhang.
\newblock ``Hodge integrals and tau-symmetric integrable hierarchies of Hamiltonian evolutionary PDEs''. 
{\em Adv. Math.} 293 (2016), 382--435.

\bibitem{DZ-norm}
B.~Dubrovin and Y.~Zhang. 
\newblock ``Normal forms of hierarchies of integrable PDEs, Frobenius manifolds and Gromov-Witten invariants''.
\newblock Preprint {\em arXiv:math/0108160}.

\bibitem{E2007}
Y.~Eliashberg.
\newblock ``Symplectic field theory and its applications''.
\newblock {\em International Congress of Mathematicians. Vol. I}, 217--246, {\em Eur. Math. Soc., Zürich}, 2007.

\bibitem{EGH2000}
Y.~Eliashberg, A.~Givental, and H.~Hofer.
\newblock ``Introduction to symplectic field theory''.
\newblock {\em Geom. Funct. Anal.} 2000, Special Volume, Part II, 560--673. 

\bibitem{FP2005}
C.~Faber and R.~Pandharipande.
\newblock ``Relative maps and tautological classes''.
\newblock {\em J. Eur. Math. Soc.} 7 (2005), no. 1, 13--49.

\bibitem{Getzler2002}
E.~Getzler.
\newblock ``A Darboux theorem for Hamiltonian operators in the formal calculus of variations''.
\newblock{\em Duke Math. J.} 111 (2002), 535--560.

\bibitem{G1994}
I.~P.~Goulden.
\newblock ``A differential operator for symmetric functions and the combinatorics of multiplying transpositions''.
\newblock {\em Trans. Amer. Math. Soc.} 344 (1994), no. 1, 421--440. 

\bibitem{JPPZ2017}
F.~Janda, R.~Pandharipande, A.~Pixton, and D.~Zvonkine.
\newblock ``Double ramification cycles on the moduli spaces of curves''.
\newblock  {\em Publ. Math. Inst. Hautes \'{E}tudes Sci.} 125 (2017), 221--266.

\bibitem{J1974}
A.-A.~A.~Jucys.
\newblock ``Symmetric polynomials and the center of the symmetric group ring''.
\newblock {\em Rep. Mathematical Phys.} 5 (1974), no. 1, 107--112.

\bibitem{KM1994}
M.~Kontsevich and Yu.~Manin.
\newblock ``Gromov-Witten classes, quantum cohomology, and enumerative geometry''.
\newblock {\em Comm. Math. Phys.} 164 (1994), no. 3, 525--562.

\bibitem{LT2001}
A.~Lascoux and J.-Y.~Thibon.
\newblock ``Vertex operators and the class algebras of symmetric groups''.
\newblock {\em Zap. Nauchn. Sem. S.-Peterburg. Otdel. Mat. Inst. Steklov. (POMI)} 283 (2001), Teor. Predst. Din. Sist. Komb. i Algoritm. Metody. 6, 156--177, 261; {\em reprinted in J. Math. Sci. (N.Y.)} 121 (2004), no. 3, 2380--2392.

\bibitem{M2015}
I.~G.~Macdonald.
\newblock ``Symmetric functions and Hall polynomials''.
\newblock Second edition. With contribution by A. V. Zelevinsky and a foreword by Richard Stanley. Reprint of the 2008 paperback edition. 
\newblock Oxford Classic Texts in the Physical Sciences. {\em The Clarendon Press, Oxford University Press, New York}, 2015.

\bibitem{M1981}
G.~E.~Murphy.
\newblock ``A new construction of Young's seminormal representation of the symmetric groups''.
\newblock {\em J. Algebra} 69 (1981), no. 2, 287--297.

\bibitem{OV1996}
A.~Okounkov and A.~Vershik.
\newblock ``A new approach to representation theory of symmetric groups''.
\newblock {\em Selecta Math. (N.S.)} 2 (1996), no. 4, 581--605.

\bibitem{R2008}
P.~Rossi.
\newblock ``Gromov-Witten invariants of target curves via symplectic field theory''.
\newblock {\em J. Geom. Phys.} 58 (2008), no. 8, 931--941.

\bibitem{T2012}
C.~Teleman. 
\newblock ``The structure of 2D semi-simple field theories". 
\newblock {\em Invent. Math.} 188 (2012), no.~3, 525--588.

\bibitem{Z2016}
D.~Zagier.
\newblock ``Partitions, quasimodular forms, and the Bloch-Okounkov theorem''.
\newblock {\em Ramanujan J.} 41 (2016), no. 1-3, 345--368. 
\end{thebibliography}
\end{document}